\documentclass[11pt]{article}
\pdfoutput=1
\usepackage{amsmath, amssymb, amsthm, amsfonts,cite,alltt,clrscode}
\usepackage[english]{babel}
\usepackage{graphicx}
\usepackage{caption}
\usepackage{url}
\usepackage{float}
\usepackage{wrapfig}
\usepackage[list=true,listformat=simple]{subcaption}

\usepackage{hyperref}

\newtheorem{theorem}{Theorem}[section]
\newtheorem{lemma}[theorem]{Lemma}

\newtheorem{corollary}[theorem]{Corollary}
\newtheorem{definition}[theorem]{Definition}

\graphicspath{{./}}

\advance \topmargin by -\headheight
\advance \topmargin by -\headsep
\evensidemargin \oddsidemargin
{\makeatletter
 \gdef\xxxmark{%
   \expandafter\ifx\csname @mpargs\endcsname\relax 
     \expandafter\ifx\csname @captype\endcsname\relax 
       \marginpar{xxx}
     \else
       xxx 
     \fi
   \else
     xxx 
   \fi}
 \gdef\xxx{\@ifnextchar[\xxx@lab\xxx@nolab}
 \long\gdef\xxx@lab[#1]#2{\textbf{[\xxxmark #2 ---{\sc #1}]}}
 \long\gdef\xxx@nolab#1{\textbf{[\xxxmark #1]}}
 \long\gdef\xxx@lab[#1]#2{}\long\gdef\xxx@nolab#1{}%
}

\let\realbibitem=\bibitem
\def\bibitem{\par \vspace{-1.2ex}\realbibitem}

\begin{document}

\title{Computational Complexity of Motion Planning \\ of a Robot through Simple Gadgets}

\author{%
  Erik D. Demaine%
    \thanks{MIT Computer Science and Artificial Intelligence Laboratory, 32 Vassar Street, Cambridge, MA 02139, USA, \protect\url{{edemaine,jaysonl}@mit.edu}, \protect\url{isaacg@alum.mit.edu}, \protect\url{mrudoy@gmail.com}}
\and
  Isaac Grosof\footnotemark[1]
    \thanks{Now at Carnegie Mellon University.}
\and
  Jayson Lynch\footnotemark[1]
\and
  Mikhail Rudoy\footnotemark[1]
    \thanks{Now at Google Inc.}
}

\date{}
\maketitle

\begin{abstract}
We initiate a general theory for analyzing the complexity of motion planning of a single robot
through a graph of ``gadgets'', each with their own state, set of locations, and
allowed traversals between locations that can depend on and change the state.
This type of setup is common to many robot motion planning hardness proofs.
We characterize the complexity for a natural simple case:
each gadget connects up to four locations in a perfect matching
(but each direction can be traversable or not in the current state), 
has one or two states, every gadget traversal is immediately undoable,
and that gadget locations are connected by an
always-traversable forest, possibly restricted to avoid crossings in the plane.
Specifically, we show that any single nontrivial four-location two-state
gadget type is enough for motion planning to become PSPACE-complete,
while any set of simpler gadgets (effectively two-location or one-state)
has a polynomial-time motion planning algorithm.
As a sample application, our results show that motion planning games with ``spinners'' are PSPACE-complete, establishing a new hard aspect of \emph{Zelda: Oracle of Seasons}.
\end{abstract}

\section{Introduction}
\label{sec:intro}

Many hardness proofs are based on \emph{gadgets} --- local pieces,
each often representing corresponding pieces of the input instance,
that combine to form the overall reduction.
Garey and Johnson \cite{NPBook} called gadgets ``basic units'' and
the overall technique ``local replacement proofs''.
The search for a hardness reduction usually starts by experimenting with small
candidate gadgets, seeing how they behave, and repeating until amassing
a sufficient collection of gadgets to prove hardness.

This approach leads to a natural question: what gadget sets suffice to
prove hardness?  There are many possible answers to this question, depending
on the precise meaning of ``gadget'' and the style of problem considered.
Schaefer \cite{Schaefer-1978-SAT} characterized the complexity of all
``Boolean constraint satisfiability'' gadgets, including easy problems
(2SAT, Horn SAT, dual-Horn SAT, XOR SAT) and hard problems
(3SAT, 1-in-3SAT, NAE 3SAT).
Constraint Logic \cite{GPCBook09} proves sufficiency of small sets of gadgets
on directed graphs that always satisfy one local rule
(weighted in-degree at least~$2$), in many game types
(1-player, 2-player, 2-team, polynomially bounded, unbounded),
although the exact minimal sets of required gadgets remain unknown.
Both of these general techniques naturally model ``global'' moves that
can be made anywhere at any time (while satisfying the constraints).
Nonetheless, the techniques have been successful at proving hardness for
problems where moves must be made local to an agent/robot that traverses the
instance.

In this paper, we introduce a general model of gadgets that naturally arises
from \emph{single-agent motion planning problems}, where a single agent/robot
traverses a given environment from a given start location to a given goal
location.  Our model is motivated by the plethora of existing hardness proofs
for such problems, such as Push-$1$, Push-$*$, PushPush, and Push-X
\cite{Demaine-Demaine-Hoffmann-O'Rourke-2003};
Push-$2$-F \cite{Demaine-Hearn-Hoffman-2002};
Push-$1$ Pull-$1$ \cite{us,Pereira-Ritt-Buriol-2016};
as well as several Nintendo video games studied at recent FUN conferences
\cite{Nintendo_TCS,Mario_FUN2016}.

\subsection{Gadget model}

In general, we model a \emph{gadget} as consisting of one or more
\emph{locations} (entrances/exits) and one or more \emph{states}.
(In this paper, we will focus on gadgets with at most two states.)
Each state $s$ of the gadget defines a labeled directed graph on the locations,
where a directed edge $(a,b)$ with label $s'$ means that the robot
can enter the gadget at location $a$ and exit at location~$b$, and that such a
traversal forcibly changes the state of the gadget to~$s'$.
Equivalently, a gadget is specified by its \emph{state space},
a directed graph whose vertices are state/location pairs,
where a directed edge from $(s,a)$ to $(s',b)$ represents that the robot
can traverse the gadget from $a$ to $b$ if it is in state~$s$,
and that such traversal will change the gadget's state to~$s'$.
Gadgets are \emph{local} in the sense that traversing a gadget does
not change the state of any other gadgets.

A \emph{system of gadgets} consists of gadgets, their initial states, and
\emph{connections} between disjoint pairs of locations (forming a matching).
If two locations $a,b$ of two gadgets (or the same gadget) are connected,
then the robot can traverse freely between $a$ and~$b$
(outside the gadgets).
(Equivalently, we can think of locations $a$ and $b$ as being identified.)
These are all the ways that the robot can move: exterior to gadgets
using connections, and traversing gadgets according to their current states.
In a \emph{puzzle}, we are given a system of gadgets, the robot starts at
a specified start location, and we want to find a sequence of moves that
brings the robot to a specified goal location.
The main problem we consider here is the obvious decision problem:
is the given puzzle solvable?

\begin{wrapfigure}{r}{1.25in}
  \centering
  \vspace*{-4ex}
  \includegraphics[width=\linewidth]{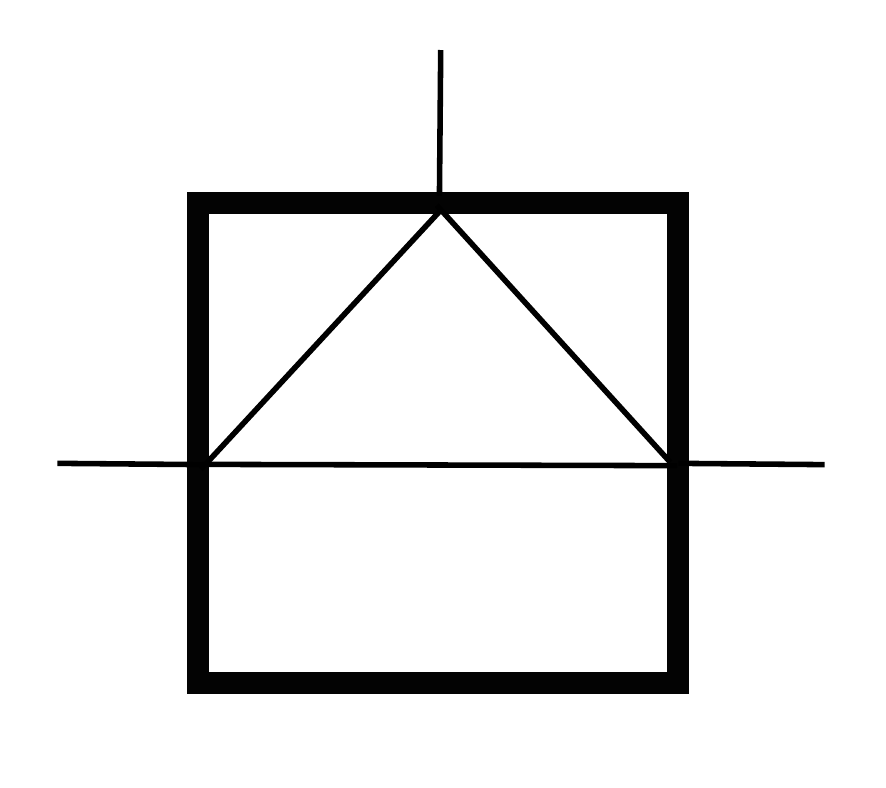}
  \vspace*{-5ex}
  \caption{Branching hallway gadget}
  \vspace*{-3ex}
  \label{branching hallway}
\end{wrapfigure}

One type of gadget we always allow in this paper
is the \textbf{branching hallway} gadget, which has one state and three
locations, and always allows traversal between all pairs of locations;
see Figure~\ref{branching hallway}.
In other words, upon reaching such a gadget, the robot is free to choose
and move to any of the three locations.
Connecting together multiple branching hallways allows us to effectively
connect the other gadgets' locations according to an arbitrary forest
(as described in the abstract).

All other gadgets we consider in this paper are ``deterministic'' and
``reversible''.
A gadget is \emph{deterministic} if its state space has maximum out-degree
$\leq 1$, i.e., a robot entering the gadget at some location $a$ in some
state $s$ (if possible) can exit at only one location~$b$ and one new state~$s'$.
A gadget is \emph{reversible} if its state space has the reverse of every edge,
i.e., it is the bidirectional version of an undirected graph.
Thus a robot can immediately undo any gadget traversal.%
\footnote{This notion is different than the sense of ``reversible''
  in reversible computing, which would mean that we could derive
  which move to undo from the current state.}
Together, determinism and reversibility are equivalent to requiring that
the state space is the bidirectional version of a matching.

Other than the (one-state) branching hallway, we further require that the
states of a gadget differ only in their orientations of the possible
traversals.  More precisely, a \emph{$k$-tunnel} gadget has $2k$ locations,
paired in a perfect matching whose pairs are called \emph{tunnels}, such that
each state defines which direction or directions each tunnel can be traversed.

We also consider \emph{planar} systems of gadgets, where the gadgets and
connections are drawn in the plane without crossings.  Planar gadgets are
drawn as small regions (say, disks) with their locations as points
in a fixed clockwise order along their boundary.  A single gadget type thus
corresponds to multiple planar gadget types, depending on the choice of the
clockwise order of locations.  Connections are drawn as paths connecting
the points corresponding to the endpoint locations, without crossing
gadget interiors or other connections.

\subsection{Our results}

We characterize the computational complexity of deciding puzzle solvability
when the allowed gadgets consist of the branching hallway and any number of
deterministic reversible $\leq 2$-state $k$-tunnel gadgets, for any $k$.
Specifically, if there is at least one gadget type that is not equivalent to
a $1$-state or $1$-tunnel gadget, then the problem is PSPACE-complete;
and otherwise, the problem is in~P.
The same characterization holds for planar systems of gadgets; thus,
in applications, we do not have to worry about building a crossover gadget
(which is often the most difficult).

In Section~\ref{sec:pspace}, we sketch our proof from \cite{us} that motion planning with two-toggle-locks and crossovers is PSPACE-complete. In Section~\ref{sec:AP2T}, we prove that one particular gadget, the antiparallel two-toggle, can simulate a variety of other gadgets, eventually including a two-toggle-lock and a crossover. As a consequence, motion planning with the antiparallel two-toggle is PSPACE-complete. In Section~\ref{sec:everything}, we show that all nontrivial deterministic reversible $2$-state, $2$-tunnel gadgets can simulate the antiparallel two-toggle. As a consequence, each corresponding motion planning problem is PSPACE-complete. In Section~\ref{sec:general}, we extend these results to give a precise hardness characterization for the motion planning problem with each deterministic reversible $2$-state $k$-tunnel gadget.

We also partially characterize the computational complexity of
deterministic reversible $\leq 2$-state gadgets with three locations.
In particular, we study spinners and deterministic forks,
as described in Section~\ref{sec:applications}.

We hope that our approach will be useful for establishing hardness of
many real-world motion planning problems and puzzles.
As sample applications, our results allow us to establish a new PSPACE-hard
aspect of the Nintendo video game \emph{Zelda: Oracle of Seasons}
(which features spinners),
and to provide alternate proofs of hardness results for Push-1 Pull-1,
as described in Section~\ref{sec:applications}.

\section{Gadget Basics}
\label{sec:basics}
To categorize the possible deterministic reversible $2$-state $2$-tunnel gadget
types, we first categorize the possible tunnel types in such a gadget.
A tunnel is \emph{trivial} if it is either never traversable or always traversable.
A trivial tunnel can always be split into a separate $1$-state $1$-tunnel
gadget, so we can ignore them.
What remain are three possible \emph{nontrivial} tunnel types:
\bigskip

\begin{tabular}{llp{0.6\textwidth}}
  $\vcenter{\hbox{\includegraphics[width=0.15\textwidth]{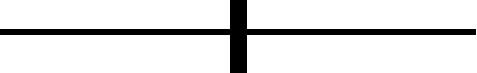}}}$ &
  \bf Tripwire &
  A tunnel that can always be traversed in either direction,
  but traversing it switches the gadget's state.
  \\[1ex]
  $\vcenter{\hbox{$\includegraphics[width=0.15\textwidth]{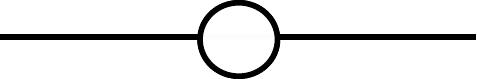}\atop
  \includegraphics[width=0.15\textwidth]{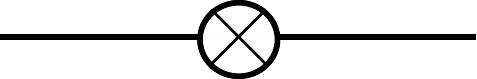}$}}$ &
  \bf Lock &
  In the \emph{unlocked} state (shown above), the tunnel can be traversed
  in either direction; in the \emph{locked} state (shown below),
  the tunnel cannot be traversed in either direction.
  \\[1ex]
  $\vcenter{\hbox{\includegraphics[width=0.15\textwidth]{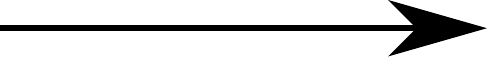}}}$ &
  \bf Toggle &
  A tunnel that can always be traversed in a single direction,
  where the direction differs in the two states of the gadget. The state is switched when the gadget is traversed.
\end{tabular}

\bigskip

There are six ways to combine these tunnel types into pairs.
Two combinations, Lock--Lock and Tripwire--Tripwire, are trivial combinations
equivalent to one-state gadgets in which each tunnel is either always
traversable in both directions or never traversable.
Thus we restrict our attention to the four other combinations, listed below. 
Because we are interested in planar systems,
we consider the multiple planar gadgets for each nontrivial combination.
(We do, however, treat a gadget and its reflection as equivalent.)
As a result, there are nine different nontrivial two-tunnel two-state gadgets,
abbreviated and listed below.
The bulk of our paper focuses on the six gadgets shown in
Figure~\ref{fig:GadgetDiagrams}, which omits most crossing variants.

\begin{enumerate}
    \item {\bf Tripwire--Lock}: Traversing the tripwire makes the other tunnel flip between being passable and impassable, causing it to `lock' or `unlock'. There are crossing and non-crossing varieties, abbreviated {\bf CWL} (crossing wire lock) and {\bf NWL} (non-crossing wire lock).
    \item {\bf Toggle--Lock}: Traversing the toggle flips the lock tunnel between being passable and impassable. Crossing the lock tunnel, by definition, does not change the state of the gadget. Notice that one direction of the toggle corresponds to an open lock and the other direction to the closed lock. There are crossing and non-crossing varieties, abbreviated {\bf CTL} (crossing toggle lock) and {\bf NTL} (non-crossing toggle lock).
    \item {\bf Tripwire--Toggle}: Here traversing either the tripwire or the toggle flips the direction of the toggle. There are crossing and non-crossing varieties, abbreviated {\bf CWT} (crossing wire toggle) and {\bf NWT} (non-crossing wire toggle).
    \item {\bf Toggle--Toggle}: Also known as a {\bf 2-toggle} \cite{us}. Traversing either toggle flips the direction of both of them. This is the only case where there are two directed tunnels, leading to three possibilities: crossing, parallel, and anti-parallel. They are abbreviated {\bf C2T} (crossing 2-toggle), {\bf P2T} (parallel 2-toggle), and {\bf AP2T} (anti-parallel 2-toggle).
\end{enumerate}

\begin{figure}[t]
  \centering
  \begin{subfigure}[b]{0.15\textwidth}
    \includegraphics[width=\textwidth]{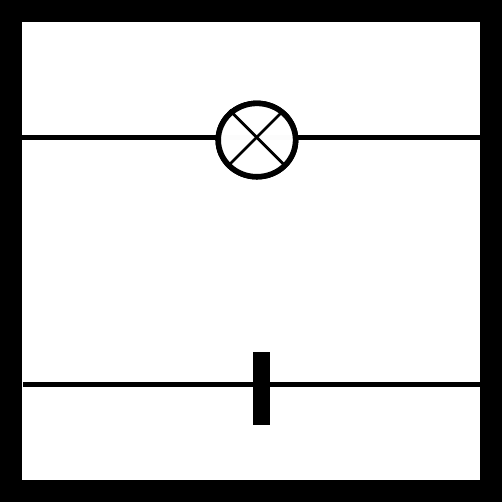}
    \caption{NWL}
    \label{fig:NWL}
  \end{subfigure}\hfill
  \begin{subfigure}[b]{0.15\textwidth}
    \includegraphics[width=\textwidth]{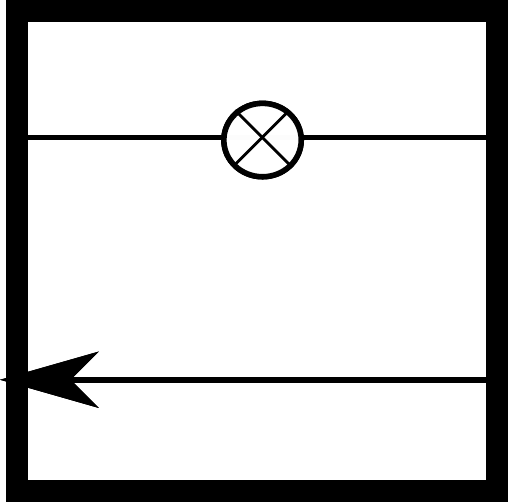}
    \caption{NTL}
    \label{fig:NTL}
  \end{subfigure}\hfill
    \begin{subfigure}[b]{0.15\textwidth}
    \includegraphics[width=\textwidth]{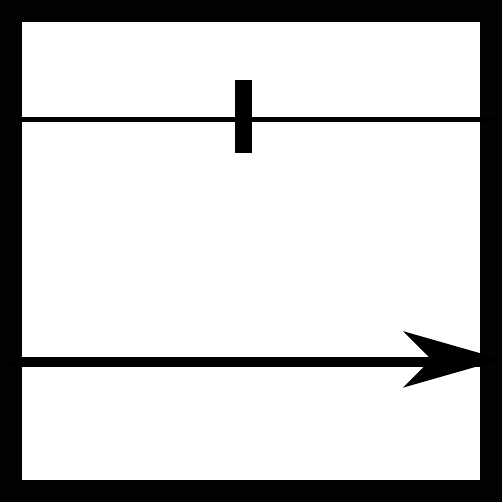}
    \caption{NWT}
    \label{fig:NWT}
  \end{subfigure}\hfill
  \begin{subfigure}[b]{0.15\textwidth}
    \includegraphics[width=\textwidth]{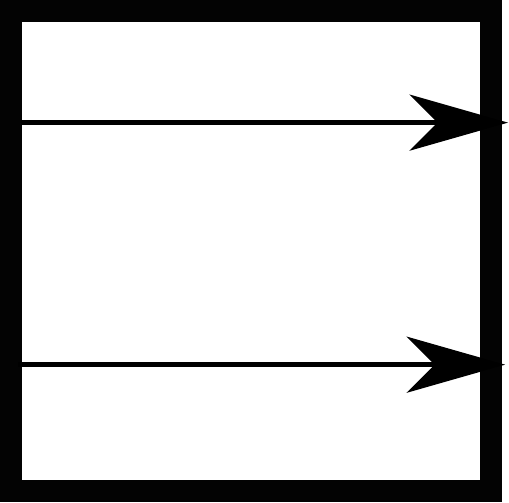}
    \caption{P2T}
    \label{fig:P2T}
  \end{subfigure}\hfill
  \begin{subfigure}[b]{0.15\textwidth}
    \includegraphics[width=\textwidth]{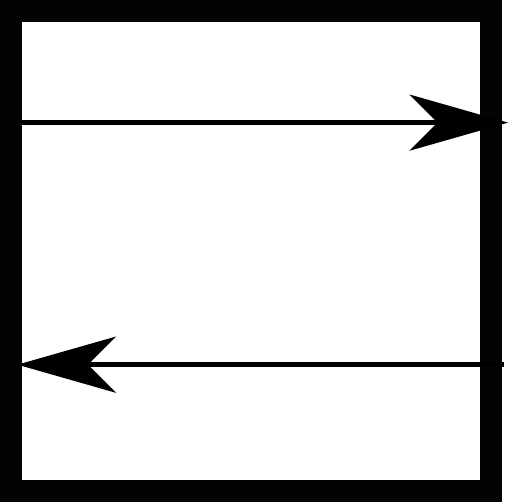}
    \caption{AP2T}
    \label{fig:AP2T}
  \end{subfigure}\hfill
    \begin{subfigure}[b]{0.15\textwidth}
    \includegraphics[width=\textwidth]{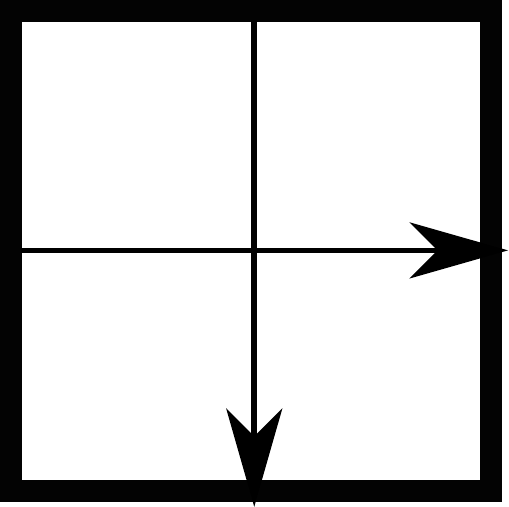}
    \caption{C2T}
    \label{fig:C2T}
  \end{subfigure}\hfill
  \caption{Six of the nine deterministic reversible 2-state gadgets on two tunnels. We leave out the CWL, CTL, and CWT gadgets as they are not heavily used in the paper.}
\label{fig:GadgetDiagrams}
\end{figure}

In this paper we will often need to discuss putting gadgets together to create new behavior. We will do so by creating a system of gadgets that is ``equivalent'' to some target gadget, thereby ``simulating'' that gadget. Two systems of gadgets are \emph{equivalent} if there is a bijective correspondence between their locations and a correspondence between their states such that the allowed transitions for all (locations, state) pairs are the same under these two correspondences. We will say that a gadget or set of gadgets \emph{simulates} a target gadget if it is possible to combine gadgets from the set (possibly using duplicates) such that the resulting system is equivalent to the target gadget. We will always implicitly allow the use of the branching hallway gadget in these constructions. In all cases, these constructions will be planar.

\subsection{Closure Properties}
\label{sec:closure}

In this section we show that systems of (deterministic) reversible gadgets remain (deterministic) and reversible. Such compositionality properties of gadgets are interesting and will simplify several later proofs.

\begin{lemma}
\label{lem:undirected_composition}
Any system of gadgets composed of two reversible gadgets is reversible.
\end{lemma}

\begin{proof}
Consider any transition through the system formed by composing two reversible gadgets. This transitions is a walk through the gadgets and connections that form a system. Since both gadgets are reversible, it is possible for the robot to enact the exact reverse of this walk after the walk is done. This will exactly reverse the effect of the walk within each gadget. Thus, it is possible to reverse the entire transition.

Since every transition of the system can be reversed, the system is reversible.
\end{proof}

Since all of the gadgets we consider in this paper are reversible, Lemma~\ref{lem:undirected_composition} means our systems will all be reversible as well. 

\begin{lemma}
  \label{lem:bijective_composition}
  Any system of gadgets composed of two deterministic reversible gadgets is deterministic and reversible.
\end{lemma}

\begin{proof}
  The state space of a reversible, deterministic gadget is an undirected matching of some (state, location) pairs to each other. This a necessary and sufficient characterization of reversible, deterministic gadgets.

  When we compose two such gadgets, we create paths through the pair of gadgets. However, no (state, location) pair has more than two edges: One connection to the other gadget, and one edge through its original gadget. Moreover, any (state, location) pair that forms an external location has a most one edge, as it does not connect to the other gadget. As a consequence, the path from any external location through the gadget is either a deterministic path to another external location, or a dead end. There is no branching, as branching would require a location with three edges.

  Thus, the resultant object is deterministic. By Lemma~\ref{lem:undirected_composition} it is reversible as well.
\end{proof}

\subsection{PSPACE Membership}

\begin{lemma}
\label{lem:in-pspace}
  Deciding puzzle solvability is in PSPACE.
\end{lemma}

\begin{proof}
The entire state of the system can be described by the current state of the gadgets and the location of the agent. The gadgets have a polynomial number of states and there can only be a polynomial number of gadgets. Since the entire state of the board fits in a polynomial amount of space, we can non-deterministically search for a solution, showing containment in NPSPACE. Savich's Theorem\cite{SAVITCH1970177} gives PSPACE${}={}$NPSPACE.
\end{proof}

\section{2-toggle-lock and crossover motion planning is PSPACE-complete}
\label{sec:pspace}
In \cite{us} we showed that motion planning with 4-toggles and crossovers is PSPACE-complete. In that construction, the crucial gadget turned out to be a 2-toggle-lock, which is a 3-tunnel, 2-state gadget with two locks and a tunnel. The 4-toggle was not used in any way after the construction of the 2-toggle-lock, showing that 2-toggle-locks and crossovers are PSPACE-hard. For convenience we sketch the proof, with some refinement. One should refer to the prior paper for a more detailed and rigorous proof.

\begin{definition}
3QSAT is the following decision problem. Given a fully quantified boolean formula in prenex normal form and in conjunctive normal form with no more than three variables per clause, decide whether the formula is true.
\end{definition}

\begin{theorem}
\label{thm:qsat}
    Motion planning with 2-toggle-locks and crossovers is PSPACE-hard.
\end{theorem}

We reduce from 3QSAT to motion-planning with 2-toggle-locks and crossovers. To do so we need to construct clauses, universal variables, and existential variables. Literals will consist of a 2-toggle-lock which will be set from the 2-toggle side and checked by passing through the lock. Clauses are composed of a branching hallway that leads through each of its associated literals.

Existential variables will be a branching hall with a group of toggle-locks in series. Passing through in one direction opens the locks of the gadgets representing true literals of that variable while closing the locks of the false ones. Going through the other way allows this to be undone, as the system is reversible.

 To construct universal quantifiers we connect up the 2-toggle sections as in Figure~\ref{fig:QSAT}, where each universal gadget consists of several antiparallel 2-toggles with locks. Each of these gadgets sends the robot forward in one state or back to the beginning in the other state, and flips the state. Repeatedly entering from the left iterates through all configurations of the states, so the robot must check all of the possible values for the universal variables. The goal state lies at the far end of the eries of universal gadgets.

For both the existentials and the universals, the variables are actually a long series of 2-toggle-locks with one lock for each literal of the variable in the formula.

When putting this all together, as in Figure~\ref{fig:QSAT}, we need to ensure that the robot cannot sneak back into the variable gadget and change existential settings it shouldn't be allowed to access, namely those existentials beyond the universal it just emerged from. To do this we construct a simple system that puts a lock on the return pathway at the end of each universal variable which only allows passage if the prior variable is set to false. Since the robot will have just exited from a variable which was set to true, this prevents the robot from moving forward in the variable chain. In addition, all earlier variables are false allowing the robot to travel back to the formula, since the universal gadgets take on incrementing binary values with each loop through the gadget. Since those existential variables are ones the robot was allowed to set to any value on the prior passage, going back and changing them now gives no advantage over having set them to that value earlier.

This safeguard is the one difference from the prior construction, which checked the values of all prior universal variables, requiring a quadratic blow-up in number of gadgets. The need for crossovers and a 2D layout will still create a quadratic blowup in problem size overall, but this simplification seemed worth noting and should allow for the 3D result to cause only a linear blowup in problem size.

With this guard in place, the robot can only reach the goal state by demonstrating a solution to the 3QSAT instance, after iterating through all settings of the universal gadget. \qed

\begin{figure}
\centering
\includegraphics[width=0.95\textwidth]{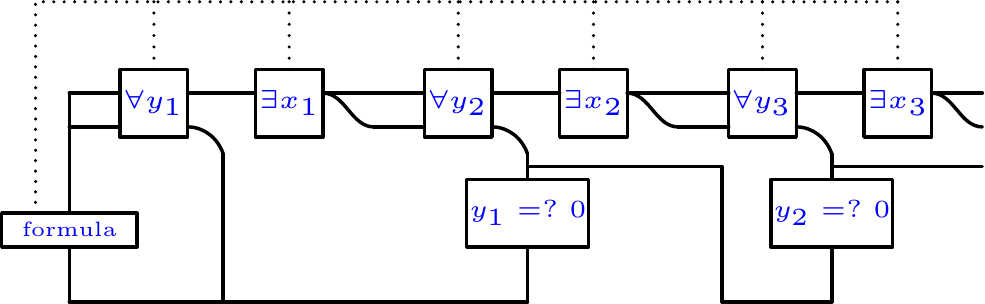}
\caption{Structure of the QSAT reduction.}
\label{fig:QSAT}
\end{figure}

\section{Antiparallel 2-toggle motion planning is PSPACE-complete}
\label{sec:AP2T}
We will show that the question of whether a robot in a system of antiparallel 2-toggle gadgets can reach a specified goal location is PSPACE-complete. To do so, we will simulate various other gadgets using AP2T gadgets, eventually simulating 2-toggle-locks and crossovers. Since motion planning with 2-toggle-locks and crossovers is PSPACE-complete, this implies that AP2T motion planning is PSPACE-complete. 

\begin{theorem}
    \label{thm:AP2T-complete}
    Motion Planning with AP2T gadgets is PSPACE-complete.
\end{theorem}

We will simulate the gadgets needed for the PSPACE-completeness proof, and a wide variety of other intermediate gadgets to help us get there. The steps are as follows:

\begin{enumerate}
    \item Simulate a C2T, using AP2Ts. Lemma~\ref{lem:C2T}.
    \item Simulate a P2T, using C2Ts. Lemma~\ref{lem:P2T}.
    \item Simulate a NTL, using AP2Ts, C2Ts and P2Ts. Lemma~\ref{lem:NTL-sim}.
    \item Simulate various types of 2-toggle locks, with ``round'' and ``stacked'' internal connections. The types of internal connections are described in Section~\ref{sec:2TL}, and the constructions are given in Lemmas \ref{lem:RP2TL} and \ref{lem:SAP2TL}.
    \item Simulate a NWL, using the stacked antiparallel 2-toggle lock. Lemma~\ref{lem:NWL}.
    \item Simulate a stacked tripwire-lock-tripwire, using NWLs. Lemma~\ref{lem:SWLW}
    \item Simulate a crossover, using stacked tripwire-lock-tripwires. Lemma~\ref{lem:crossover}
\end{enumerate}

With a 2-toggle lock and a crossover constructed, we can apply Theorem~\ref{thm:qsat} to show that motion planning with AP2Ts is PSPACE-hard. Adding in Lemma~\ref{lem:in-pspace}, we find that it is PSPACE-complete.\qed

\begin{lemma}
    \label{lem:C2T}
  Antiparallel 2-toggles (AP2Ts) simulate a crossing 2-toggle (C2T).
\end{lemma}

\begin{proof}
The construction is given in Figure~\ref{fig:AP2T-to-C2T}. In the state of the construction shown in the figure, there are two possible transitions: the robot can move from the upper left to the bottom right of the construction, or from the upper right to the bottom left. Either of those transitions toggles both AP2Ts, leaving the construction mirrored top to bottom. Thus, the construction has two states. The possible traversals in one state (as shown above) are from the top left to the bottom right and from the top right to the bottom left, while the possible traversals in the other state are (by symmetry) from the bottom left to the top right and from the bottom right to the top left. Following any of these traversals swaps the state of the construction. Notice that this is exactly the behavior of a C2T.

If the robot enters the construction shown from the upper left, upon reaching the center the robot can only proceed to the bottom right, or come back the way it came. Therefore, the upper left to bottom right transition is the only possible transition from that location. By symmetry, the same is true from top left to bottom right. Thus, the one traversal described for each location in each state is the only one possible.
\end{proof}

\begin{figure}
  \centering
  \begin{minipage}{0.45\textwidth}
    \centering
    \includegraphics[width=0.7\textwidth]{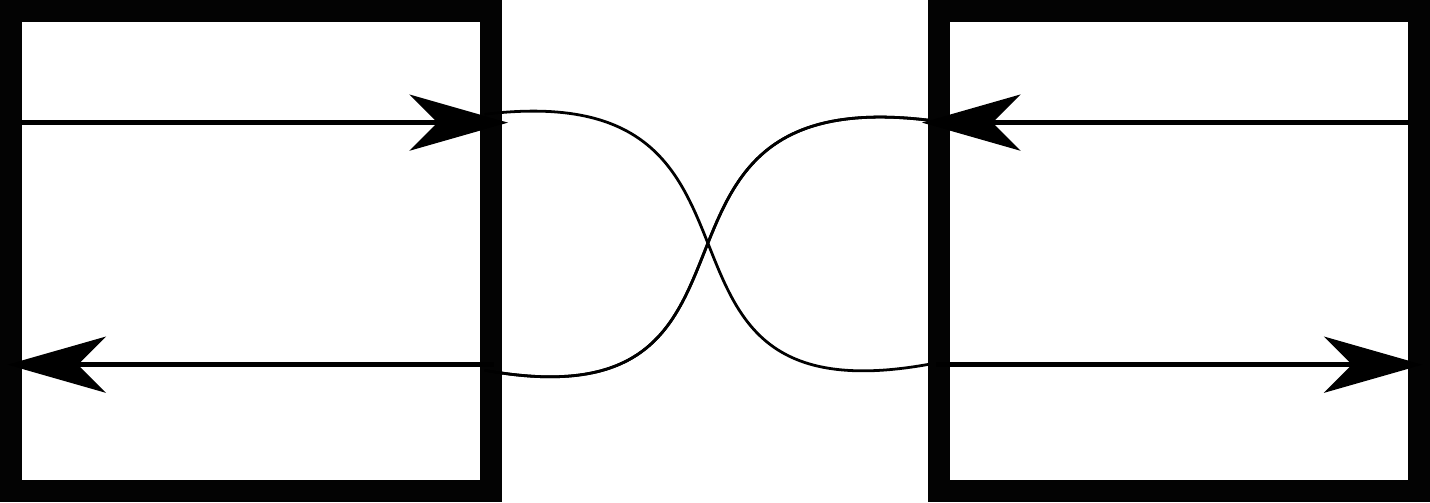}
    \caption{Anti-parallel 2-toggles simulate a crossing 2-toggle}
    \label{fig:AP2T-to-C2T}
    \centering
    \includegraphics[width=0.7\textwidth]{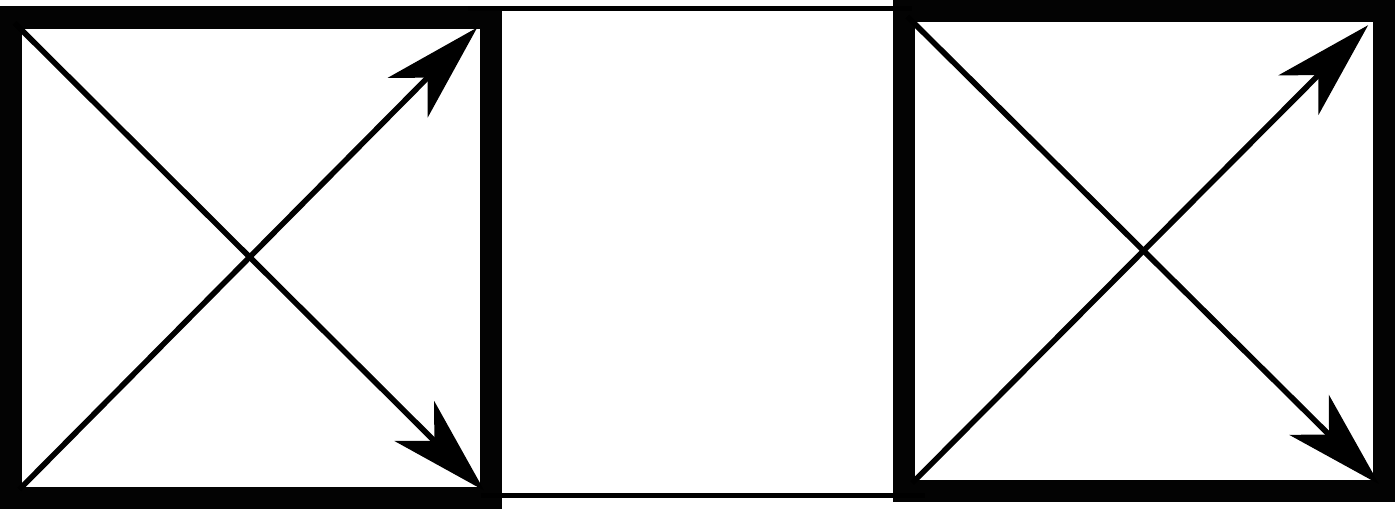}
    \caption{Crossing 2-toggles simulate a parallel 2-toggle}
    \label{fig:C2T-to-P2T}
  \end{minipage}\hfil\hfil
  \begin{minipage}{0.45\textwidth}
    \centering
    \includegraphics[width=0.9\textwidth]{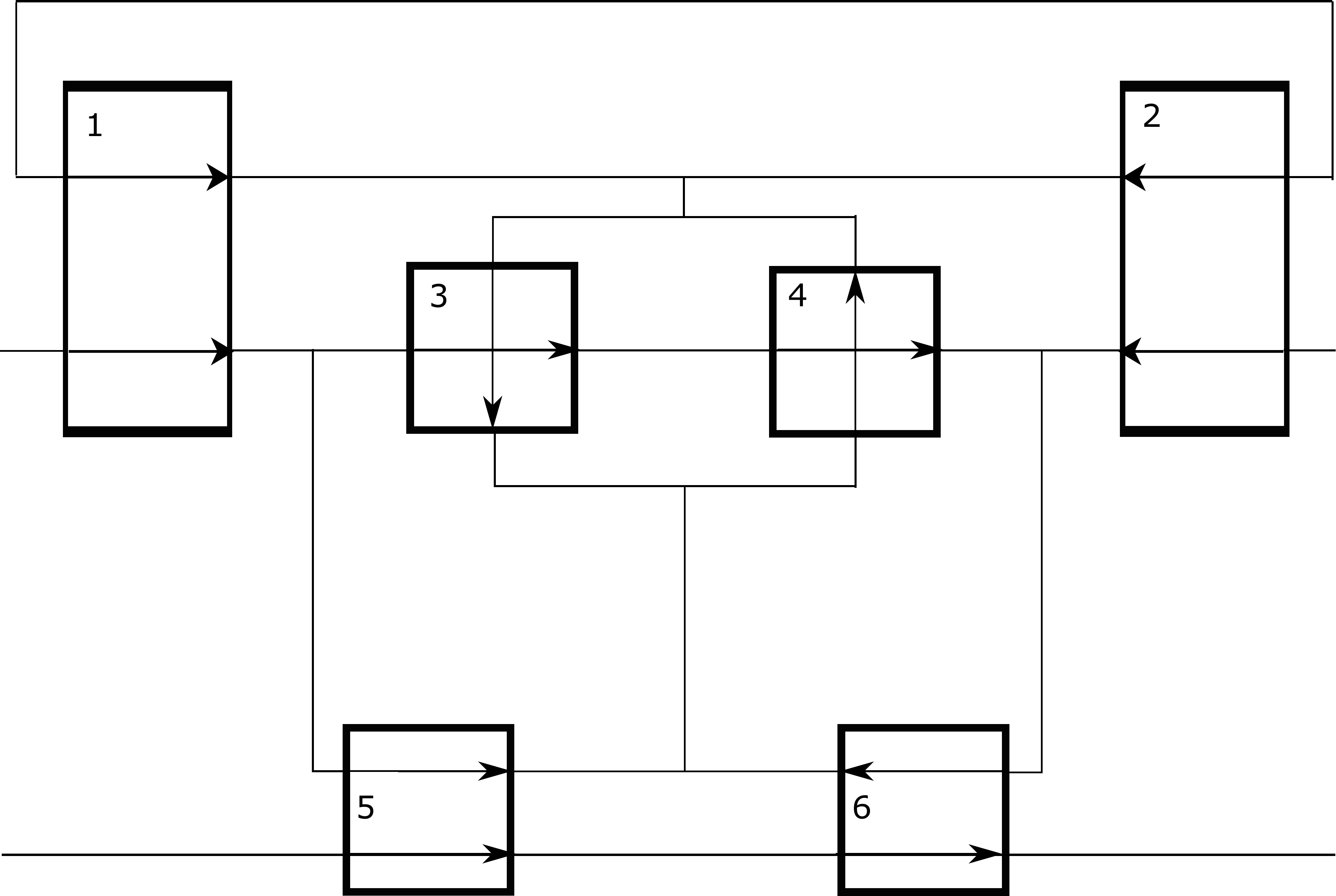}
    \caption{2-toggles simulate 1-toggle-lock.}
    \label{fig:2-toggle-to-1-toggle-lock}
  \end{minipage}
\end{figure}

\begin{lemma}
    \label{lem:P2T}
  Crossing 2-toggles (C2Ts) simulate a parallel 2-toggle (P2T).
\end{lemma}

\begin{proof}
        The construction is given in Figure~\ref{fig:C2T-to-P2T}. In the state of the construction shown in the figure, there are two possible transitions: the robot can move from the top left to the top right of the construction, or from the bottom left to the bottom right. Either of these transitions toggles both C2Ts, leaving the construction mirrored left to right. The allowed traversals in one state (as shown above) are from the top left to the top right and from the bottom left to the bottom right, while the allowed traversals in the other state are (by symmetry) from the top right to the top left and from the bottom right to the bottom left. Following any of these traversals swaps the state of the construction. Notice that this is exactly the behavior of a P2T.
        
        Since the system is composed entirely of C2Ts (without even branching hallways), which are both reversible and deterministic, the result is also both reversible and deterministic, by Lemma~\ref{lem:bijective_composition}. Thus, the one transition described for each location in each state is the only transition possible.
\end{proof}

\begin{lemma}
    \label{lem:NTL-sim}
  2-toggles (AP2Ts, P2Ts and C2Ts) simulate a noncrossing toggle lock (NTL).\end{lemma}

\begin{proof}

The construction is shown in Figure~\ref{fig:2-toggle-to-1-toggle-lock}.

In this lemma, we will refer to toggles 1 and 2 in the figure as the ``outer toggles'', toggles 3 and 4 as the ``middle toggles'', and toggles 5 and 6 as the ``bottom toggles''. We will call the pathway through the lower tunnels of the bottom toggles the ``bottom tunnel'' of the overall gadget, and the rest of the gadget the ``middle tunnel'' of the overall gadget. 

An NTL has two externally observable states: locked, and unlocked.
        The locked state corresponds to the upper tunnels of the bottom toggles oriented out,
        and the unlocked state corresponds to the bottom toggles oriented in. The unlocked state is shown in Figure~\ref{fig:2-toggle-to-1-toggle-lock}.

        In this gadget, there are two internal states corresponding to each
        external state: with the horizontal tunnels of the middle toggles both oriented left, and with both oriented right.
        The only accessible states of this gadget are the states with the outer toggles oriented in, the middle toggles
        oriented both left or both right, and upper pathways of the bottom toggles
        oriented both in or both out. We will show that the gadget allows
        exactly the traversals of the NTL from these configurations,
        and cannot be left in any other configuration.

        The bottom tunnel traversals are straightforward --- the bottom tunnel
        acts as a toggle, and a traversal flips both bottom toggles, and hence the externally observable state.

        Also clearly, the robot cannot move between the bottom tunnel and the middle tunnel.

        Now, we wish to establish that in the unlocked state, the robot can always traverse
        the middle tunnel in either direction.
        In the state shown, the middle tunnel may be traversed from external location to external location as follows:

        \begin{itemize}
            \item
        The robot can get across, left to right, by traversing the following toggles
        in the following order: enter through toggle 1's lower tunnel, down to toggle 5,
        up to toggle 4's vertical tunnel, through toggle 1's upper tunnel, around the top to toggle 2's top tunnel, back down through toggle 4, back out through toggle 5, across through toggle 3's horizontal tunnel, then through toggle 4's horizontal tunnel, then out through toggle 2's lower tunnel.
            \item
        The robot can get across, right to left, by traversing the following toggles
        in the following order: enter through toggle 2's lower tunnel, down to toggle 6, up to toggle 4's vertical tunnel, through toggle 2's top tunnel, around to toggle 1's top tunnel, down through toggle 3's vertical tunnel, back out through toggle 6, across through toggle 4's horizontal tunnel, then through toggle 3's horizontal tunnel, then out through toggle 1's lower tunnel.
    \item
        
        If the middle toggles are in the opposite orientation, the system is simply
        mirrored, left to right, and the traversals are still possible.
\end{itemize}

        Next, we wish to establish that the robot cannot cross the middle tunnel in the locked state.
        After entering from either middle tunnel location, the only
        traversable toggles are the middle toggles. After traversing those,
        the robot can go no further. The bottom toggles can't be traversed, so the
        entire middle region is inaccessible. As a consequence,
        the opposite outer toggle's upper pathway can't be accessed.
        Therefore the robot can only leave via its original location.

        We also must establish that if the gadget starts in one of the configurations mentioned,
        the robot must leave it in the proper state, and can't leave it in a configuration that wasn't mentioned. This is straightforward for the bottom tunnel, so we will focus on the middle two locations.

        We will show that the accessible configurations of the gadget are exactly as described. To do so, we will make use of the concept of a cut in a gadget.

\begin{lemma}
\label{lem:cut_rule}
Let $A$ be a connected region of a planar embedding of a gadget system which does not contain any locations. Then the boundary of $A$, which we will call a cut, is traversed an even number of times during any traversal of the construction.
\end{lemma}
\begin{proof}
Whenever the boundary of $A$ is crossed, the robot goes from inside $A$ to outside or vice versa. Since the robot starts a traversal outside $A$ and ends it outside $A$, it must cross the boundary an even number of times.
\end{proof}

        The upper pathways of the outer toggles form a cut, and the lower pathways of the
        outer toggles form a cut. 
        Thus, the upper pathways of the outer toggles are crossed an even
        number of times, and the lower pathways are passed an even number
        of times, so the outer toggles must be passed an even number of times in total.
        Thus, the toggles must either be both oriented in or both out when leaving.
        However, when leaving the gadget, the outer toggle which the robot exited
        through must end up oriented in,
        so both outer toggles must end up oriented in.

        The vertical pathways of the middle toggles form a cut. The horizontal 
        pathways form a cut. Thus, upon leaving, the middle toggles must have
        been traversed an even number of times in total,
        and hence must end up both left or both right.

        The upper pathways of the bottom toggles must be passed an even number
        of times. So the upper pathways of those toggles must either be both
        in or both out when leaving the gadget system.

        Thus, the gadget system must be left in a state where the outer toggles are oriented in,
        the middle toggles are oriented either both left or both right, and the upper pathways
        of the bottom toggles are oriented either both in or both out. Therefore, these are exactly the accessible configurations, as desired.

        Finally, we show that the robot leaves the gadget in the same state it was entered in, if it is entered on the middle tunnel. If the robot passes through one of the upper tunnels of the bottom toggles, when it leaves the region bounded by the bottom toggles' upper tunnels, it must leave one of the bottom toggle's upper tunnels oriented in. By the parity constraint, both bottom toggles' upper tunnels will be oriented in, thus leaving the gadget in the unlocked state. If the central tunnels are entered in the unlocked state, they will be left in the unlocked state. In the locked state, the upper tunnels of the bottom toggles cannot be passed, and so the gadget will be left in the locked state.

        Thus, the construction correctly simulates a NTL.
\end{proof}

\subsection{2-toggles and non-crossing toggle locks simulate 2-toggle locks}
\label{sec:2TL}
\begin{wrapfigure}{r}{1.5in}
  \centering
  \vspace*{-2ex}
  \includegraphics[width=\linewidth]{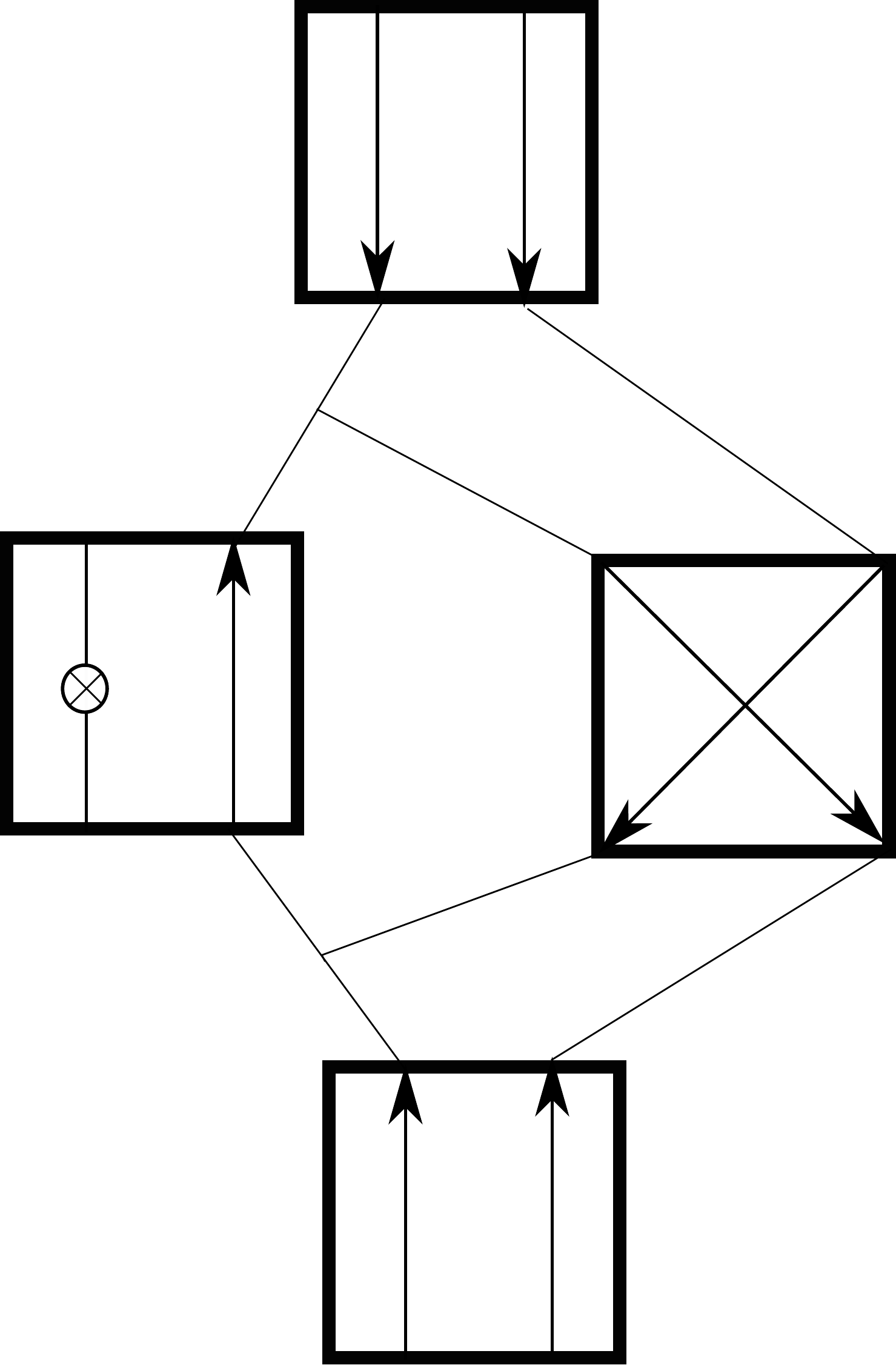}
  \caption{Round antiparallel 2-toggle-lock construction}
  \vspace*{-3ex}
  \label{fig:2-toggle-lock}
\end{wrapfigure}

We introduce some new three tunnel objects. There are several distinct planar topologies of the tunnels in a three tunnel object. We will focus on the two topologies which can be drawn with no internal crossing tunnels: three tunnels around the perimeter, and three tunnels in parallel. We will call the former a ``round'' topology, and the latter a ``stacked'' topology. Note that in the stacked topology, the order of the tunnels is relevant. In either topology, if there are multiple toggles, the relative orientation must still be specified.

        \begin{lemma}
            \label{lem:RP2TL}
          2-toggles and noncrossing toggle locks simulate a round antiparallel 2-toggle-lock (RAP2TL) and a round parallel 2-toggle-lock (RP2TL).
        \end{lemma}
        \begin{proof}

        The construction shown in Figure~\ref{fig:2-toggle-lock} simulates the behavior of a round antiparallel 2-toggle-lock.
        It has two externally accessible states: as shown, and with the middle two gadgets flipped. These correspond to the 2-toggle of the RAP2TL being pointed counterclockwise and clockwise respectively.

        We will demonstrate that this gadget is equivalent to a RAP2TL by examining all possible traversals. From the two locations that are on the lock tunnel of the NTL, the only possible traversals are to each other, if the lock tunnel is unlocked. This forms the lock tunnel of the RAP2TL.

        Traversals from the top left location: The robot must go down and to the right, due to the orientation of the toggle of the NTL.
        Then, the robot can go through the C2T, at which point
        it is blocked by the orientation of the bottom P2T. Thus, no traversal is possible from this location in this state.

        Traversals from the top right location: The robot can go through the C2T, then through the NTL. At this point, the robot cannot go through the C2T again, because the C2T has
        been toggled. Therefore, its only option is to go through the upper P2T and leave at the top left location. This traversal toggles both of the middle two gadgets, and toggles the upper P2T twice. Thus, the external state of the gadget is flipped. This is the equivalent of traversing the upper toggle of the RAP2TL that we are simulating.

        Traversals from the bottom left location: The robot must go up and to the left, due to the orientation of the C2T. Then, the robot can go through the NTL. Due to the orientation of the upper P2T, the robot must now go through the C2T. Now, the robot can leave at the bottom right location. This traversal toggles both of the middle two gadgets, and toggles the lower P2T twice. Thus, the external state of the gadget is flipped. This is the equivalent of traversing the lower toggle of the RAP2TL that we are simulating. 

        Traversals from the bottom right location: The robot is blocked by the orientation of the C2T. Thus, no traversal is possible from this location in this state.

        The opposite state is equivalent to a top-bottom mirror reversal, except for a change in the state of the lock, which does not affect which traversals are possible. Thus, in every state, this system of gadgets is equivalent to a round antiparallel two-toggle-lock (RAP2TL).
        
        Consider the gadget which is the same as the one in Figure~\ref{fig:2-toggle-lock}, except that the bottom P2T  is replaced with a C2T with its toggles allowing traversals from the bottom locations into the gadget. Clearly, the effect of this change is to swap the roles of the bottom two locations. As a result, this new construction is a round parallel two-toggle-lock, a RP2TL.
        \end{proof}


\begin{figure}
  \centering
  \begin{minipage}{0.45\textwidth}
    \centering
    \includegraphics[width=0.7\textwidth]{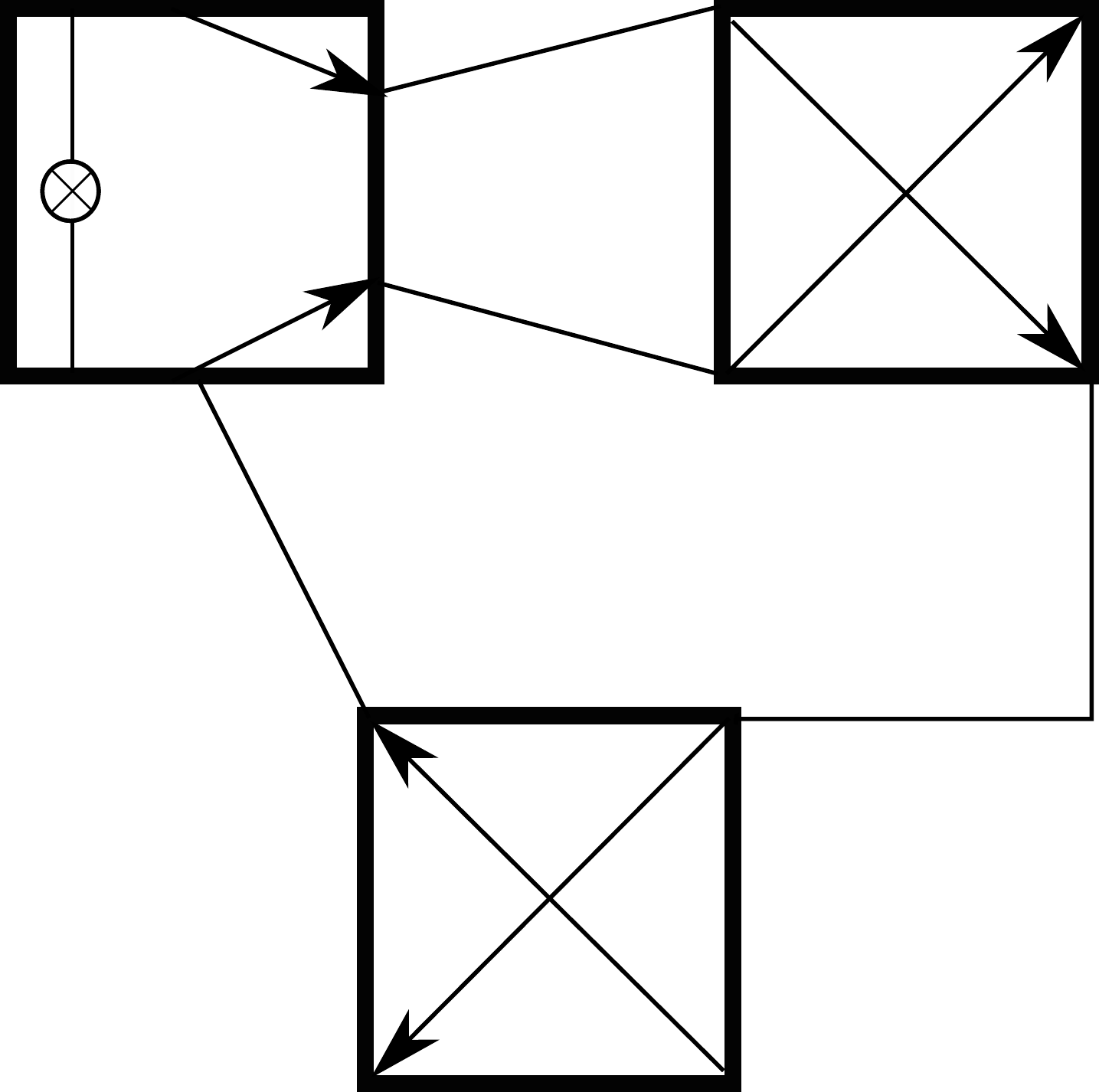}
    \caption{A round parallel 2-toggle lock is used to construct a stacked antiparallel 2-toggle lock}
    \label{fig:RP2TL-to-SAP2TL}
  \end{minipage}\hfil\hfil
  \begin{minipage}{0.45\textwidth}
    \centering
    \includegraphics[width=\textwidth]{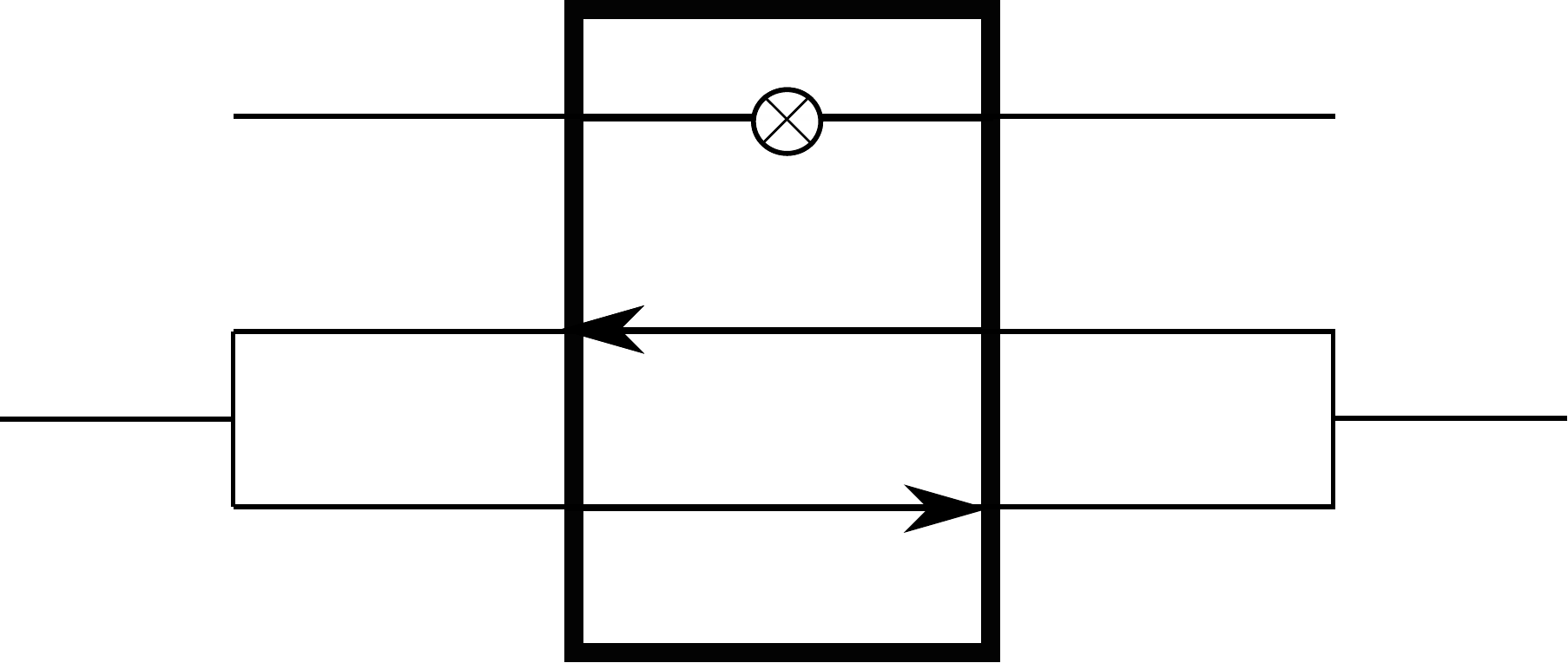}
    \caption{A noncrossing tripwire lock constructed from an anti-parallel 2-toggle and lock with the lock on the side}
    \label{fig:2-toggle-lock-to-NWL}
  \end{minipage}\hfill
\end{figure}

\begin{lemma}
    \label{lem:SAP2TL}
  RP2TLs and 2Ts simulate a stacked antiparallel 2-toggle-lock (SAP2TL).
\end{lemma}
\begin{proof}
  A SAP2TL is a three tunnel gadget where the three tunnels cross the gadget in parallel, with the two antiparallel toggle tunnels next to each other.
  
  Starting with a RP2TL and two C2Ts, we can simulate a SAP2TL as shown in Figure~\ref{fig:RP2TL-to-SAP2TL}. The lock tunnel is straightforward. The two other traversals are from the top left to the bottom left, and from the bottom right to the top right. Both of these traversals pass through every gadget. In the other state, all three gadgets are flipped, and the same traversals are possible in the opposite direction.

  Since every state-affecting traversal traverses all gadgets, the states of the three gadgets always switch together, and the behavior is that of an SAP2TL. Equivalently, by Lemma~\ref{lem:bijective_composition}, the system of gadgets is deterministic and reversible, so the three traversals mentioned are the only ones possible, and the construction simulates a SAP2TL.
\end{proof}

\subsection{2-toggle locks simulate non-crossing wire locks}
\begin{lemma}
    \label{lem:NWL}
  AP2TLS simulates a NWL.
\end{lemma}
\begin{proof}
  By connecting the locations of the SAP2TL as shown in Figure~\ref{fig:2-toggle-lock-to-NWL}, we can simulate a NWL.

  Each traversal of either connected toggle tunnel flips the state. The connections between these two tunnels ensure that travel in either direction is always possible. As a result, the combination of these connected pathways acts as a tripwire, always allowing the robot to pass in either direction and opening or closing the lock with each traversal.
  \end{proof}
     
\subsection{Non-crossing wire locks simulate crossovers}
\begin{figure}
    \centering
  \includegraphics[width=.8\textwidth]{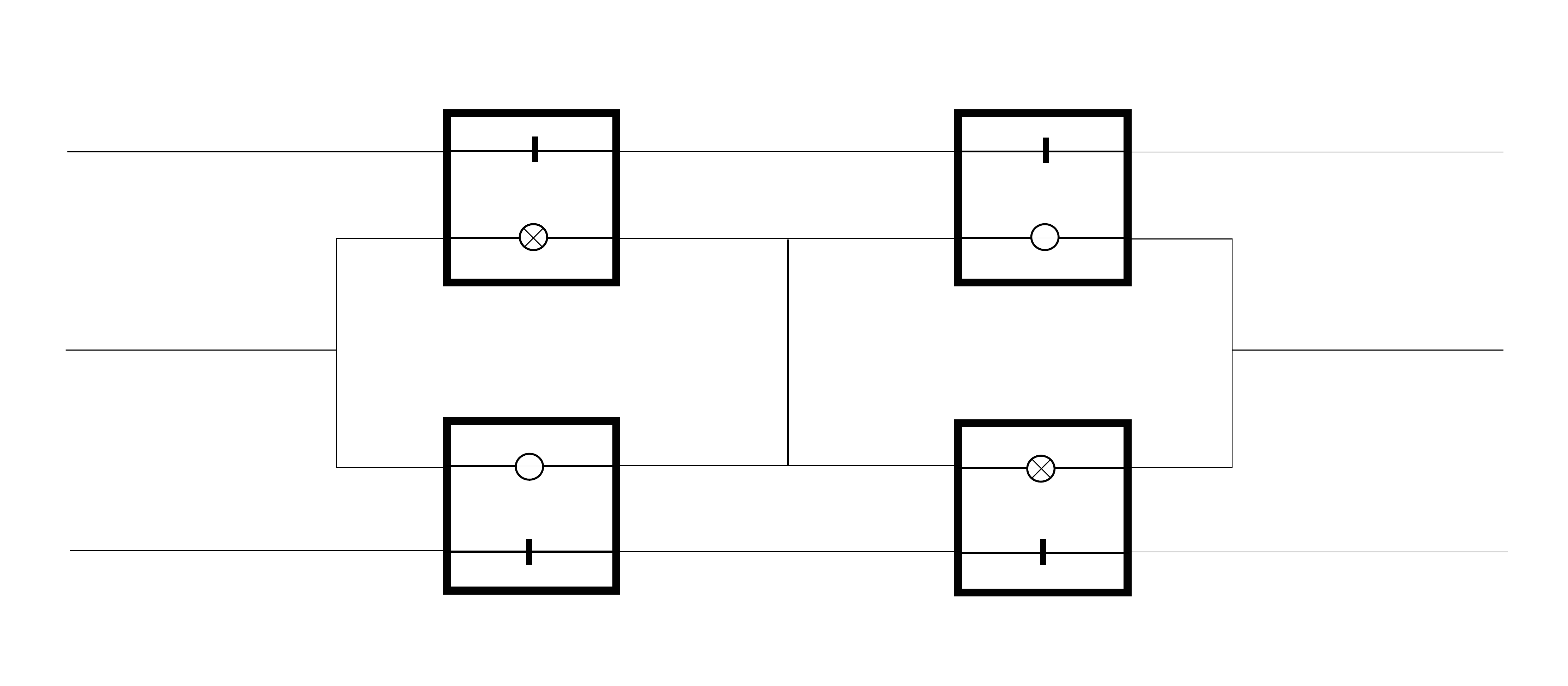}
  \caption{A stacked tripwire-lock-tripwire constructed from non-crossing tripwire locks.}
  \label{fig:SWLW}
\end{figure}
\begin{figure}
  \centering
  \includegraphics[width=.8\textwidth]{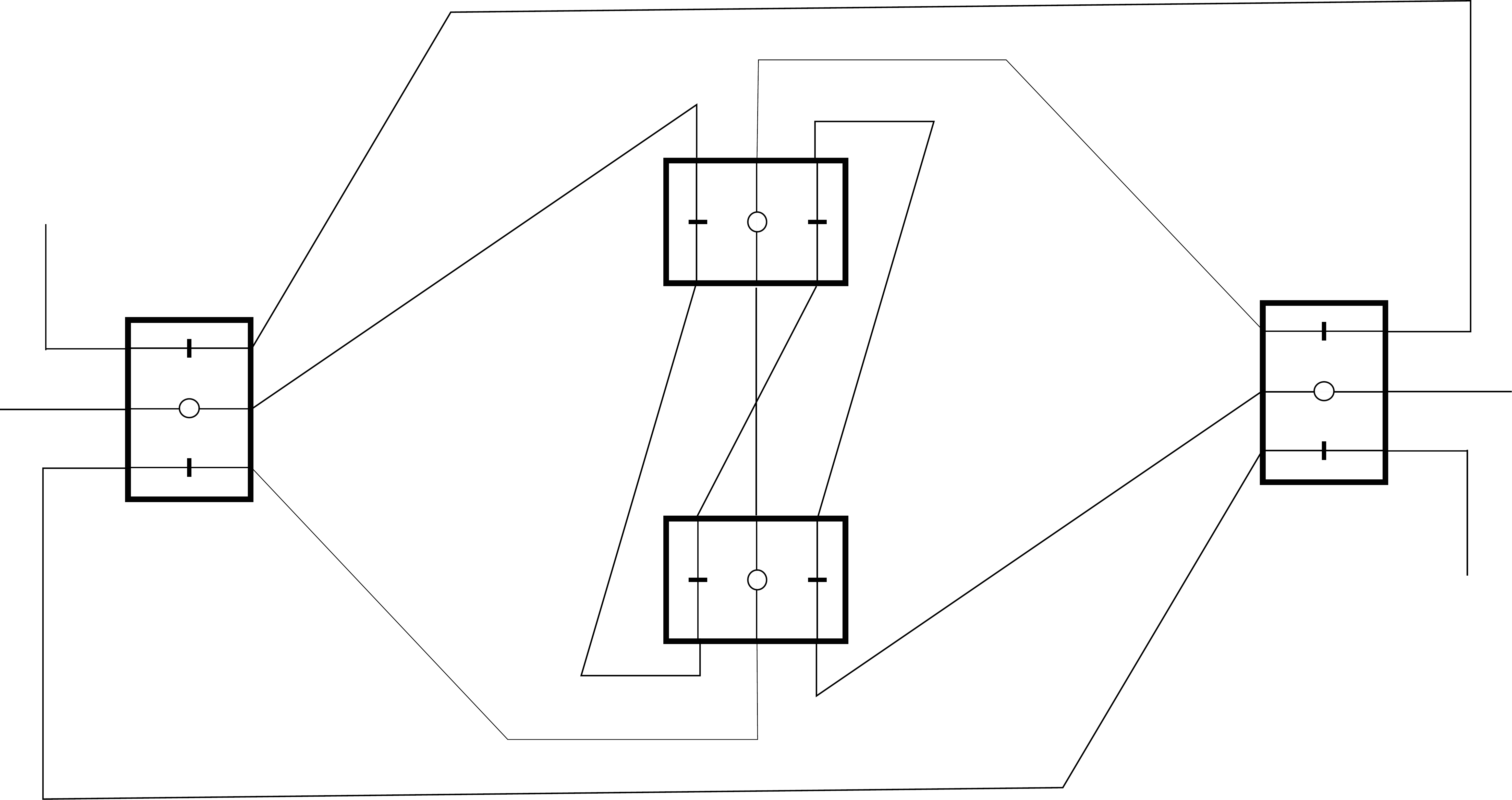}
  \caption{A crossover constructed from stacked tripwire-lock-tripwires}
  \label{fig:crossover}
\end{figure}
On our way to simulating a crossover, we will simulate another three tunnel gadget, a stacked tripwire-lock-tripwire (SWLW). Note that the lock tunnel is specifically the center tunnel.

\begin{lemma}
  \label{lem:SWLW}
  NWLs simulate a stacked tripwire-lock-tripwire (SWLW).
\end{lemma}
\begin{proof}
  The construction is shown in Figure~\ref{fig:SWLW}.
There are four accessible states of this gadget, which are any of the states where there is one locked and one unlocked NWL among the two top NWLs, and one of each among the two bottom NWLs.

The states can only be changed by traversing the tripwire tunnels, and doing so flips both NWLs on the side traversed, maintaining the invariant.

If both left NWLs are locked, or both right NWLs are locked, the center tunnel is not passable. In the other two accessible states, the center tunnel is passable. The two pairs correspond to the two external states, with the lock locked and unlocked respectively. In any state, traversing either tripwire moves the gadget to a state with the opposite passability of the lock tunnel. Thus, this construction simulates a SWLW.
\end{proof}
\begin{lemma}
	\label{lem:crossover}
  SWLWs simulate a crossover.
\end{lemma}
\begin{proof}
        \label{sec:crossover}
        The gadget shown in Figure~\ref{fig:crossover} implements a crossover. The robot may always cross from left to right,
        right to left, top to bottom and bottom to top, but in no other directions. There is a single accessible state, the one with all four SWLWs in the unlocked state.

        When the robot enters from any of the four external locations it has only a single option up until the point where it reaches the four-way intersection at the center. Upon reaching this point, the robot has traversed the tripwire tunnels of two of the SWLWs, locking them. In particular, the SWLWs whose lock tunnels are on the two orthogonal pathways are locked. For instance, if the robot entered from the top, the left and right pathway's SWLWs would be locked at this point. As a result, the only way for the robot to continue is to go straight, passing through the other tripwires of the same two SWLWs, and emerging from the other side. The robot has completed a crossover traversal, with no other options.

        Because the robot passed through the tripwires of two SWLWs twice, and only the lock tunnels of the other two SWLWs, the object is left in its original state, making the state shown in Figure~\ref{fig:crossover} the only accessible state. This construction correctly simulates a crossover.
\end{proof}

For the PSPACE-completeness result, we make use of 2-toggle locks and crossovers. Combining the lemmas in Section~\ref{sec:AP2T}, we have the result we will make use of:

\begin{theorem}
  \label{thm:AP2T-all}
  AP2Ts simulate crossovers and all 2-toggle-locks.
\end{theorem}
\begin{proof}
  By composing the lemmas in Section~\ref{sec:AP2T}, we see that AP2Ts simulate crossovers and RAP2TLs. By using the crossover to effectively rearrange locations, we can simulate an arbitrary 2-toggle-lock.
  \end{proof}

\section{Everything simulates everything else}
\label{sec:everything}
The remaining gadgets of interest are each individually (when combined with branching hallways) sufficient to make motion planning problems PSPACE-complete. Moreover, each gadget can be simulated by a constant number of each other gadget. To prove this, we give simple gadgets to show how to construct noncrossing-tripwire-toggles from anti-parallel-2-toggles, and anti-parallel 2-toggles from each of noncrossing-toggle-locks, noncrossing-wire-locks, noncrossing-wire-toggles and parallel-2-toggles. We then show that a crossing version of a gadget can very simply make a non-crossing version of the same gadget.

\begin{theorem}
    \label{thm:simulate}
    The 2-toggles, toggle-locks, tripwire-locks and tripwire-toggles, in all orientations, can each simulate each other.
\end{theorem}
\begin{proof}
  We have already established that AP2Ts can simulate P2Ts, C2Ts, NTLs and NWLs and crossovers. We will establish that:
  \begin{itemize}
      \item AP2Ts can simulate NWTs. Lemma~\ref{lem:NWT}.
      \item P2Ts, NTLs, NWTs and NWLs can each simulate AP2Ts. Lemmas \ref{lem:P2T-A}, \ref{lem:NTL-A}, \ref{lem:NWT-A}, \ref{lem:NWL-A}, respectively.
      \item C2Ts can simulate P2Ts by Lemma~\ref{lem:P2T}, and hence AP2Ts as well.
      \item CTLs can simulate NTLs, CWLs can simulate NWLs, and CWTs can simulate NWTs. Lemma~\ref{lem:uncross}.
  \end{itemize}
  Thus, every gadget can simulate AP2Ts, and AP2Ts can simulate every non-crossing gadget, as well as crossovers. By combining non-crossing gadgets with crossovers, AP2Ts can simulate every gadget. This gives a simulation of every gadget by every other gadget, via AP2Ts as an intermediate step.
\end{proof}

\begin{corollary}
    \label{thm:all-complete}
Motion planning with any one of the gadgets in Theorem~\ref{thm:simulate} (and branching hallways) is PSPACE-complete.
\end{corollary}
\begin{proof}
Corollary~\ref{thm:all-complete} follows from Theorem~\ref{thm:simulate}, which establishes that each gadget can simulate a AP2T, and Theorem~\ref{thm:AP2T-complete}, which establishes that motion planning with AP2Ts is PSPACE-complete.
\end{proof}

\begin{lemma}
    \label{lem:NWT}
	AP2Ts simulate an NWT.
\end{lemma}

\begin{proof}
\begin{figure}
\centering
\includegraphics[width=0.5\textwidth]{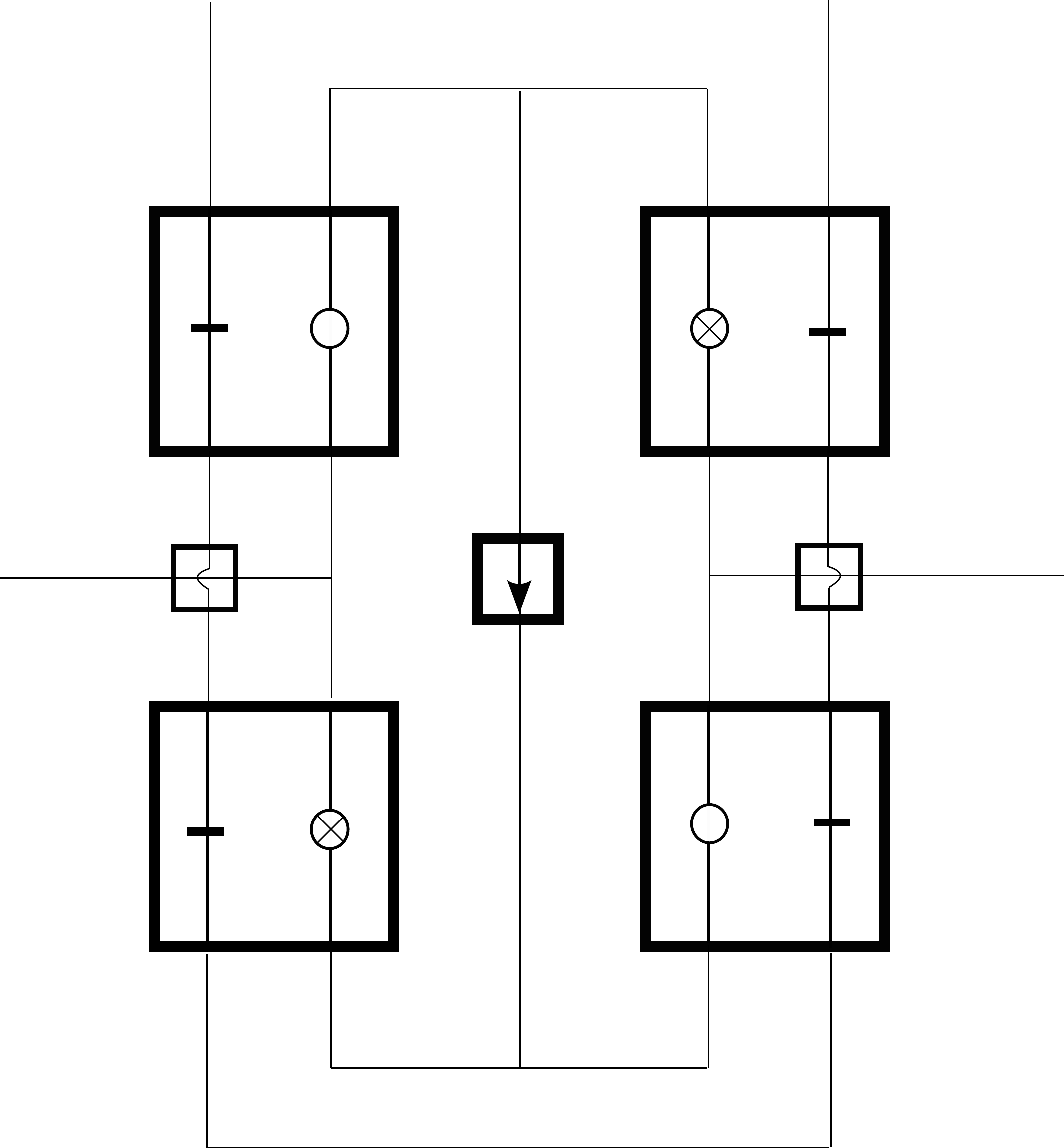}
\caption{A noncrossing wire toggle constructed from a toggle, four noncrossing tripwire locks, and two crossovers.}
\label{fig:NWL-to-NWT}
\end{figure}

We will construct a NWT as shown in Figure~\ref{fig:NWL-to-NWT}. This requires NWLs, crossovers, and 1-toggles. We already have existing constructions of NWLs and crossovers with AP2Ts. We can also build a 1-toggle with an AP2T simply by ignoring one of the two tunnels. Thus, all that's left is to show that the construction successfully simulates a NWT.

There are four accessible states: As shown in Figure~\ref{fig:NWL-to-NWT}, with all of the NWLs flipped, with the toggle flipped, and with everything flipped. The first and last correspond to the external state where the toggle is pointed right, while the other two correspond to the external state where the toggle is pointed right. The horizontal tunnel corresponds to the toggle, while the U-shaped tunnel corresponds to the tripwire in the composed gadget. In the state shown in the figure, the toggle is oriented to the right from the external perspective.

    Clearly, traversing the U-shaped tunnel will flip all of the tripwires of the NWL, resulting in a state which corresponds to the opposite external state, as desired.

    In the state shown in the figure, the horizontal tunnel may be traversed from left to right along a unique pathway due to the placement of the locks, flipping the toggle along the way. The orientation of the toggle blocks the right to left traversal. Thus, in this state, the upper tunnel may be traversed in one direction resulting in an allowed state which corresponds to the opposite external state, as desired.

    Placing the toggle in the opposite state is equivalent to a rotation by $\pi$ of the upper tunnel, showing this state also correctly simulates an NWT.

    Flipping the states of all of the NWLs is equivalent to a vertical reflection of the upper tunnel, showing this state also correctly simulates an NWT.
\end{proof}

\begin{lemma}
    \label{lem:P2T-A}
	P2Ts simulate an AP2T.
\end{lemma}

\begin{proof}
\begin{figure}
\centering
\includegraphics[width=0.7\textwidth]{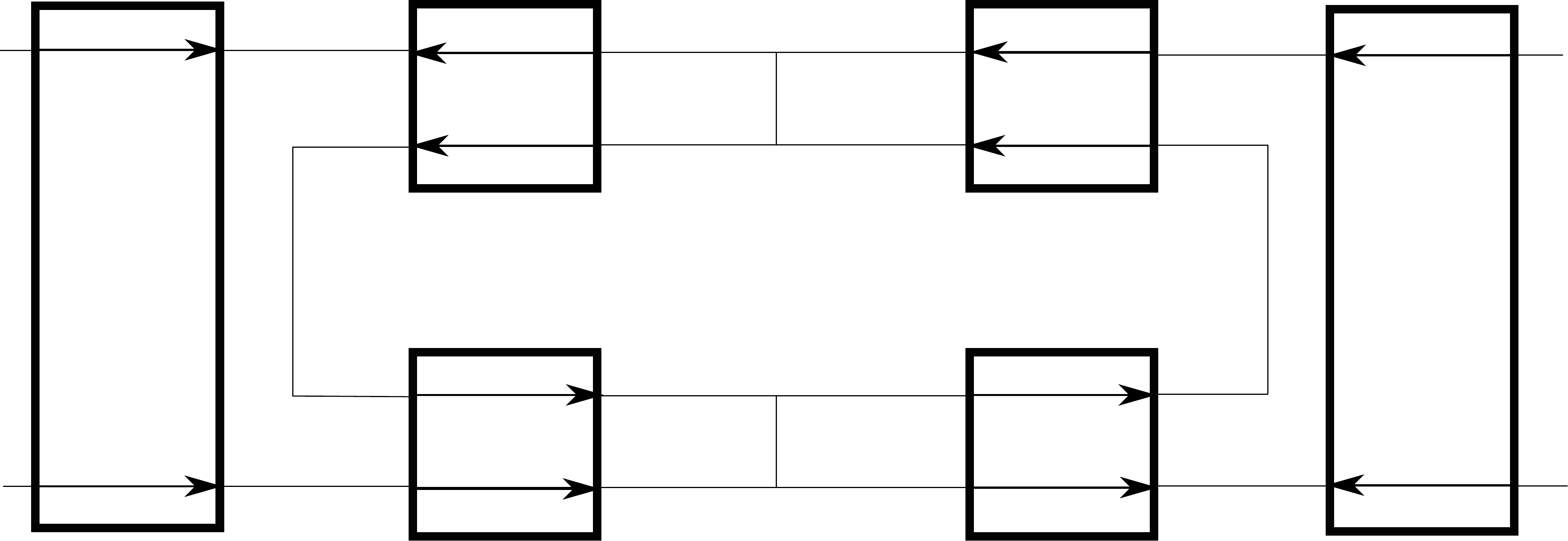}
\caption{Parallel 2-toggles simulate anti-parallel 2-toggles}
\label{fig:P2T-to-AP2T}
\end{figure}

  Figure~\ref{fig:P2T-to-AP2T} gives a construction of an antiparallel-2-toggle out of parallel-2-toggles.

There are two accessible states: As shown, and with the four inner P2Ts flipped. The former corresponds to the AP2T having a tunnel connecting the left two locations with its toggle oriented upward, and a tunnel connecting the right locations with its toggle oriented downward, while the latter corresponds to the two toggles flipped.

First, let us examine the bottom right location in the state shown in the figure. After passing the rightmost P2T, the robot is blocked. No transitions or state changes are possible. This matches the desired behavior, because the right toggle in the AP2T being simulated is oriented down.

Next, let us examine the top right location in the state shown in Figure~\ref{fig:P2T-to-AP2T}. After passing the rightmost P2T, then the upper right P2T, the robot may now either proceed along the top tunnel, or down to the central loop. In the former case, the robot may pass through the upper left P2T, but then is blocked. In the later case, the robot may either proceed around the loop to the left or to the right. If the robot goes to the right, it can pass through the lower tunnel of the upper right P2T, but then is stuck. If the robot goes to the left, it can pass through the lower tunnel of the upper left P2T, then the upper tunnel of the lower left P2T.

At this point, the robot may either continue around the loop, or exit the loop downward. If the robot continues around the loop, it can pass through the upper tunnel of the lower right P2T, but then is stuck. If it exits the loop, it can either go left or right on the bottom tunnel. If it goes left, it can pass through the lower tunnel of the lower left P2T, but then is stuck. If it goes right, it can pass through the lower tunnel of the lower right P2T, then the lower tunnel of the rightmost P2T, and exit the gadget.

Overall, we observe that the robot can make exactly one transition, from top right to bottom right. The right toggle is traversed twice, and the inner toggles are all traversed once, leaving the gadget in the other accessible state. No other transition or state change is possible, from that entrance.

Since the gadget is rotationally symmetric about its center, the possible transitions from the right mirror the possible transitions from the left. Since the other state is simply the state shown in the figure mirrored top-to-bottom, the transitions described mirror the transitions in the other state as well.
\end{proof}

\begin{lemma}
	\label{lem:NTL}
    \label{lem:NTL-A}
	NTLs simulate an AP2T.
\end{lemma}

\begin{proof}
\begin{figure}
\centering
\includegraphics[width=0.7\textwidth]{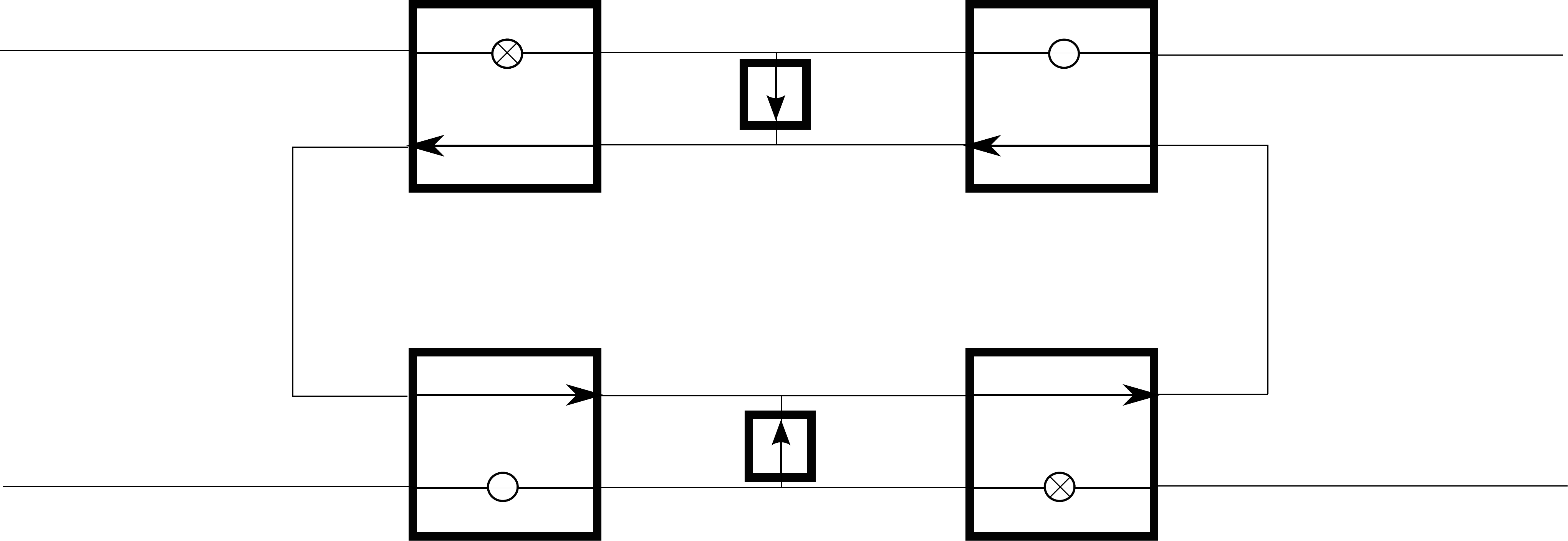}
\caption{Noncrossing-toggle-lock simulates anti-parallel-2-toggle}
\label{fig:1TL-to-AP2T}
\end{figure}

The construction is shown in Figure~\ref{fig:1TL-to-AP2T}. The two accessible states are the state shown in the figure and the state with all of the NTLs flipped, but the one-toggles still oriented inward. These correspond to an AP2T with the top tunnel directed left and bottom tunnel directed right, and the left-right mirror image.

If the robot enters from the top right, after passing the lock of the top right NTL, it must pass the upper one-toggle and proceed into the central loop. Since the lower toggle is directed upward, the robot must eventually leave the central loop via the upper toggle. The robot may now proceed around the loop. The loop may only be traversed counterclockwise, and it may only be traversed once. The robot may of course backtrack at any point, but when it leaves via the upper toggle, it must have either traversed the loop zero or one times. In the former case, the robot must leave via the top right location, leaving the system in its original state. In the latter case, the robot must leave via the top left location, as all of the locks have flipped. Thus, the top tunnel may be traversed via a right to left traversal, flipping the state, and that is the only traversal in that direction.

If the robot enters from the top left, it is immediately blocked by the lock, and no traversal is possible. Thus, the top tunnel works as desired.

Since the gadget possesses rotational symmetry around its center, the bottom tunnel is exactly the same, allowing only a left to right traversal, flipping the state.

The opposite state is the same as the original state except for a left-right right mirror reversal, so it also functions exactly as desired from the AP2T.
\end{proof}

\begin{lemma}
    \label{lem:NWT-A}
	NWTs simulate an AP2T.
\end{lemma}

\begin{proof}
A noncrossing wire toggle can simulate an anti-parallel 2-toggle with the simple construction shown in Figure~\ref{fig:NWT-to-AP2T}. The direction of each tunnel is dictated by the toggle on the tunnel, and the wire ensures both toggles are synchronized. Thus when either tunnel is traversed, both NWTs flip and the direction each tunnel can be traversed flips.
\end{proof}
\begin{figure}
\centering
\includegraphics[width=0.7\textwidth]{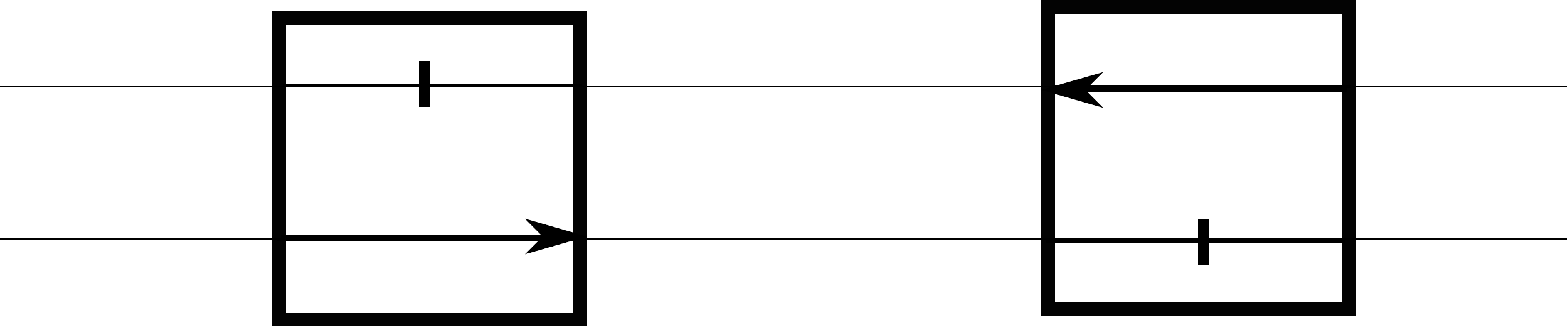}
\caption{Noncrossing-wire-toggle simulates anti-parallel-2-toggle}
\label{fig:NWT-to-AP2T}
\end{figure}

\begin{lemma}
    \label{lem:NWL-A}
	NWLs simulate an AP2T.
\end{lemma}

\begin{proof}
The construction of an anti-parallel 2-toggle from non-crossing tripwire locks can be seen in Figure~\ref{fig:NWL-to-AP2T}. Note that a 1-toggle can be constructed from an NWL by simply connecting one location of the wire to one location of the lock. A closed lock will prevent travel in one direction, but crossing the tripwire in the other direction will open the lock and allow the robot to proceed. An open lock will allow travel in the other direction. In the direction starting from the tripwire, the tripwire will close the lock in front of the robot preventing traversal. In either traversal, the tripwire is crossed, flipping the state.

	There are two main parts to this gadget, the top and bottom tunnels, and the inner loop. As with the NTL construction from Lemma~\ref{lem:NTL}, the 1-toggles ensure that the loop must be exited from the same place it was entered, which ensures all gadgets on the loop are traversed the same number of times. Since all wires are on this loop, in a given traversal of this gadget system, all of the NWLs will change state the same number of times, keeping them in sync. The upper and lower paths each contain a locked and unlocked tunnel. The locked portion prevents entry and interaction with the gadget. From the unlocked side, the robot is able to enter the gadget and flip its state an arbitrary number of times. If the state is flipped an even number of times, the robot's only path out is the way it came. If an odd number of flips have occurred, the robot can now exit through the opposite side of its path, leaving the gadget in the opposite state.

	Therefore, the gadget may traversed right to left along the top tunnel, flipping the state, and left to right along the bottom tunnel, flipping the state. We have built an AP2T.
	\end{proof}
\begin{figure}
\centering
\includegraphics[width=0.7\textwidth]{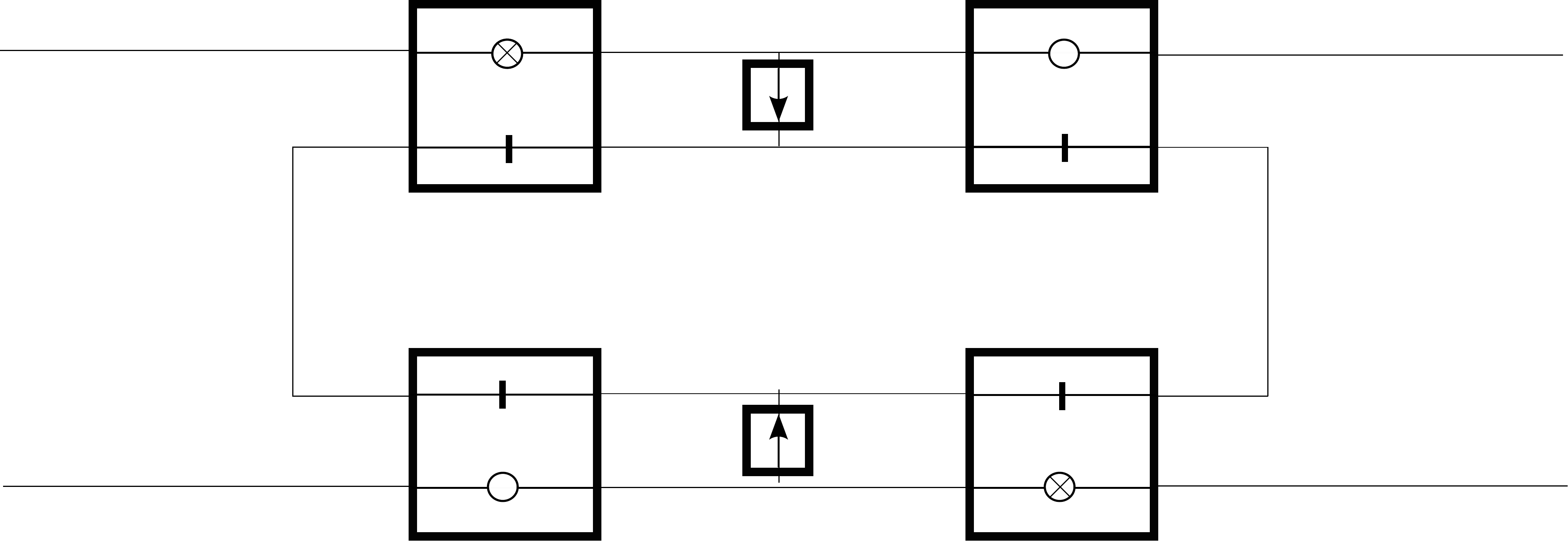}
\caption{Noncrossing-wire-lock simulates anti-parallel-2-toggle}
\label{fig:NWL-to-AP2T}
\end{figure}

\begin{lemma}
    \label{lem:uncross}
  CWTs simulate an NWT, CWLs simulate an NWL, CTLs simulate an NTL.
\end{lemma}

\begin{wrapfigure}{R}{3in}
  \centering
  \includegraphics[width=\linewidth]{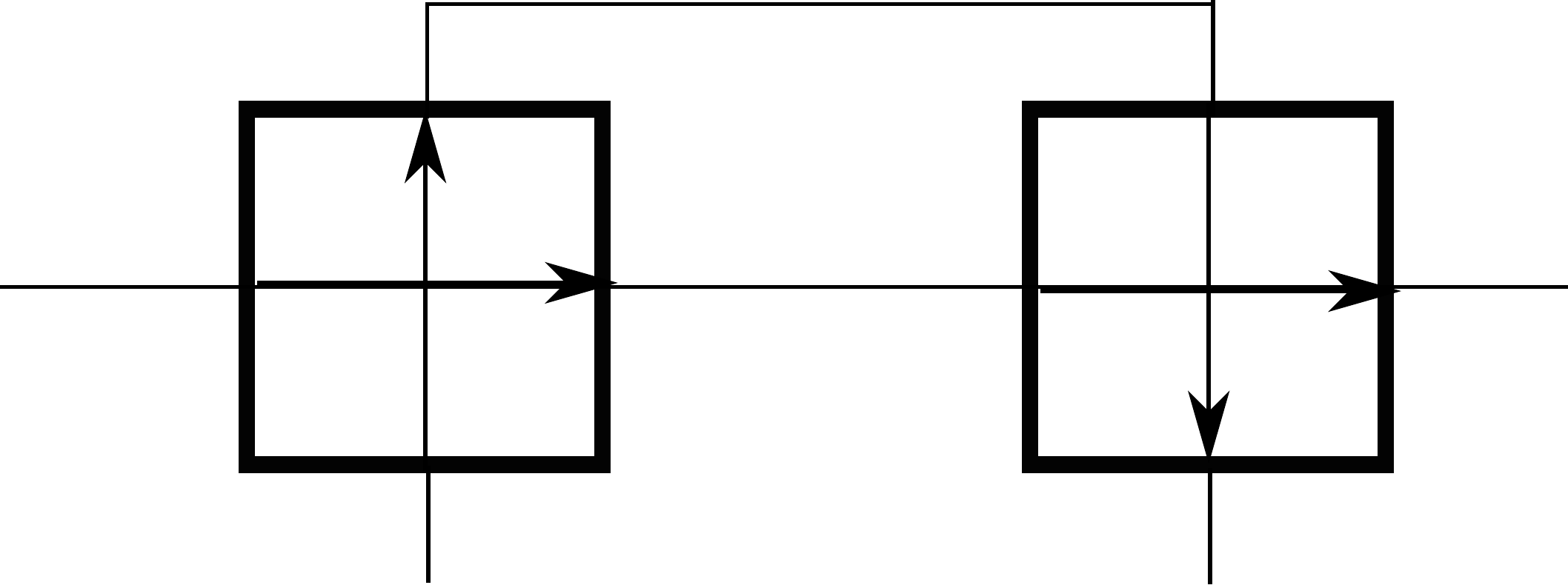}
  \caption{Crossing 2-toggles simulate parallel 2-toggle}
  \label{fig:C2T-to-N2T}
  \vspace{-5mm}
\end{wrapfigure}

  In general, one can very easily simulate a non-crossing version of a 2-tunnel gadget from the crossing version. Figure~\ref{fig:C2T-to-N2T} shows a parallel-2-toggle being constructed from a crossing-2-toggle. The same construction works for uncrossing the other gadgets we have analyzed, namely tripwire-toggles, tripwire-locks and toggle-locks. Going from non-crossing to crossing versions is significantly more complicated (except in the case of anti-parallel-2-toggle to crossing-2-toggle) but we are rescued from the need of such constructions by being able to simulate a general crossover in Lemma~\ref{lem:crossover}.

\section{Applications}
\label{sec:applications}
In this section, we give two examples of actual puzzle games
where our motion planning framework applies, establishing PSPACE-hardness.

\subsection{Spinners}
\label{sec:spinners}
A $k$-spinner is a two state deterministic reversible gadget on $k$ locations. In one state, each location is connected to its neighbor by a directed edge in a clockwise direction. In the other state, all locations are likewise connected in a counterclockwise direction. A $4$-spinner is shown in Figure~\ref{fig:zelda_spinner}. The study of $4$-spinners was posed by Jeffrey Bosboom due to their appearance in The Legend of Zelda: Oracle of Seasons. We show that for any $k\geq4$, path-planning problems with $k$-spinners and branching hallways is PSPACE-complete.

First, we can take a $k$ spinner and have all but three consecutive locations lead to dead ends. The remaining three locations form a gadget that we call a deterministic fork. A deterministic fork is a reversible, deterministic gadget on three locations. In one state, it allows the robot to go from the center to the right location and return from the left to the center location. In the other state these directions are reversed. Figure~\ref{fig:4-spinner-to-2-toggle-with-spinner} shows the construction of a crossing 2-toggle from two $4$-spinners or equivalently two deterministic forks.

\begin{figure}
  \centering
  \begin{minipage}{0.48\textwidth}
    \centering
    \includegraphics[width=0.7\textwidth]{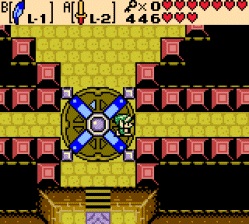}
    \caption{Example of a 4-spinner in The Legend of Zelda: Oracle of Seasons.}
    \label{fig:zelda_spinner}
  \end{minipage}\hfill
  \begin{minipage}{0.48\textwidth}
    \centering
    \includegraphics[width=0.7\textwidth]{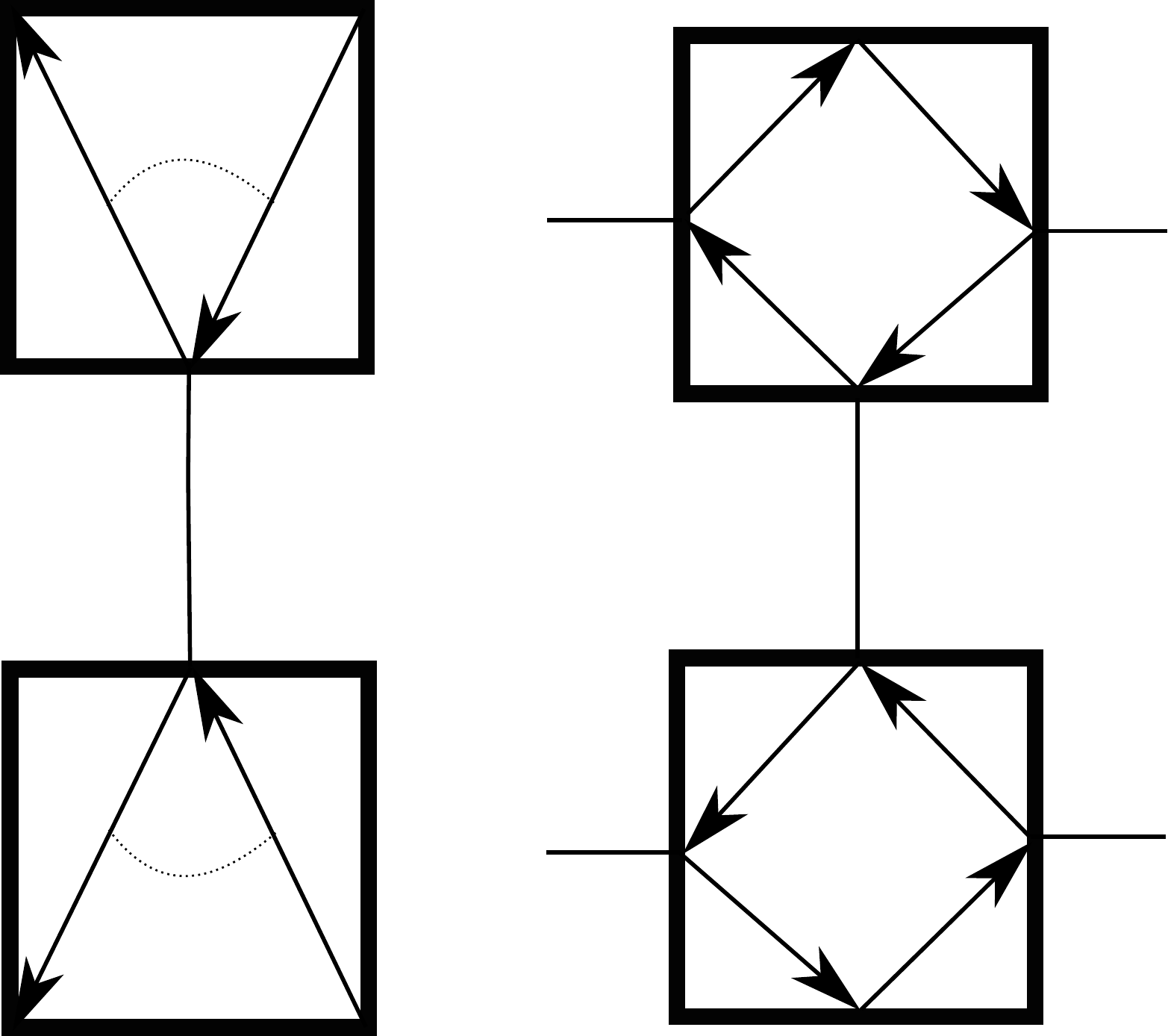}
    \caption{4-spinners simulate deterministic forks which simulate crossing 2-toggles}
    \label{fig:4-spinner-to-2-toggle-with-spinner}
  \end{minipage}
\end{figure}

\begin{theorem}
  \label{thm:k-spinners}
    For any $k\geq4$, the path-planning problem with $k$-spinners and branching hallways is PSPACE-complete.
\end{theorem}
\begin{proof}
We construct a deterministic fork by ignoring $k-3$ of the edges in the spinner. Two deterministic forks together simulate a crossing 2-toggle as shown in Figure~\ref{fig:4-spinner-to-2-toggle-with-spinner}. By Corollary~\ref{thm:all-complete}, the motion planning problem with crossing 2-toggles is PSPACE-complete.
\end{proof}

\begin{corollary}
Determining if a player can beat a level in generalized The Legend of Zelda: Oracle of Seasons is PSPACE-hard.
\end{corollary}
\begin{proof}
The Legend of Zelda: Oracle of Seasons contains $4$-spinners and requires the player to navigate from one location to a target location in a grid. Since planar graphs can be laid out in a grid with only quadratic blowup \cite{layout}, we can reduce from motion planning problems with $4$-spinners which are PSPACE-complete by Theorem~\ref{thm:k-spinners}.
\end{proof}

The complexity of motion planning with $3$-spinners, as well as the two other reversible, deterministic, 2 state, 3 location gadgets, remains open. Since $2$-spinners are the same as an edge in a graph, this would give a tight characterization for the spinner gadget. The authors would also be interested to know what other games and puzzles use spinners.

\subsection{3D Push-1 Pull-1 Block Puzzles}

\begin{wrapfigure}{r}{2.2in}
\centering
\vspace*{-10ex}
\includegraphics[width=\linewidth]{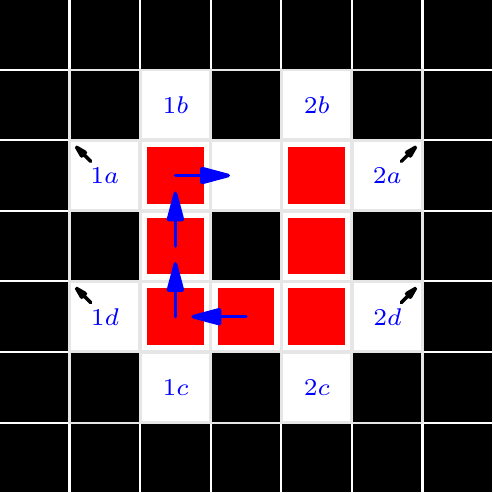}
\vspace*{-3ex}
\caption{2-toggle gadget in 3D Push-1 Pull-1 with fixed blocks.}
\vspace*{-6ex}
\label{fig:push-pull-2-toggle}
\end{wrapfigure}

Our interest in reversible gadgets originated in the study of Push-Pull
Block Puzzles, specifically, Push-1 Pull-1 where the robot can push a
single $1 \times 1$ block on a square grid,
and similarly can pull, the exact reverse of a push operation \cite{us}.
Push-1 Pull-1 has since been proved
PSPACE-complete, even in 2D \cite{Pereira-Ritt-Buriol-2016}.
Nonetheless, our motion-planning framework applies at least to 3D
Push-1 Pull-1 with fixed blocks, providing an alternate proof of
PSPACE-hardness.  Specifically, Figure~\ref{fig:push-pull-2-toggle}
shows how to build a 2-toggle gadget (using just three vertical layers).

\section{General hardness characterization}
\label{sec:general}

  Here, we tightly characterize the hardness of the motion planning problem with all deterministic, reversible, $2$-state, $k$-tunnel gadgets.
  \begin{theorem}
    \label{thm:general}
    Motion planning with any deterministic, reversible, $2$-state, $k$-tunnel planar gadget (with branching hallways) is PSPACE-complete if and only if the gadget has two toggle tunnels, a toggle tunnel and a tripwire tunnel, a toggle tunnel and a lock tunnel or a tripwire tunnel and a lock tunnel. Motion planning with all other such gadgets is in P.
  \end{theorem}
  
First, we provide upper bounds for some classes of simpler gadgets. This shows that, for their category, our hardness results are minimal in the sense that path planning with simpler gadgets in the same class can be solved in P.

\begin{theorem}
    \label{thm:1-state}
Gadgets with only one state are in NL.
\end{theorem}
\begin{proof}
One state gadgets cannot change in any way. Thus they must all be comprised of static descriptions of allowed traversals from one location to another. This can be modeled as a mixed graph. Path planning in mixed graphs is in NL\cite{SAVITCH1970177}.
\end{proof}

The only nontrivial gadget on 1 tunnel with two states which is reversible and deterministic is the 1-toggle.
\begin{theorem}
\label{thm:1-toggle}
Motion planning with 1-toggles is in NL.
\end{theorem}
\begin{proof}
We reduce this problem to ST connectivity in mixed graphs. To solve this problem we simply treat every 1-toggle as a directed edge pointed in the direction the 1-toggle is initially oriented and then run the standard algorithm. It is obvious that if a solution here exists then a path in the 1-toggle planning problem also exists. What is less clear is that this is sufficient to find any such path.

Consider a path which traverses at least one toggle more than once. Consider the last toggle on the path which is traversed more than once. After this toggle is traversed, only toggles which are traversed at most once are on the path. Call this toggle $t$, and let its final traversal be from $u$ to $v$. Since $t$ was traversed repeatedly, there was some previous point in the path where the robot was at $v$, before it traversed $t$ the second-to-last time. Let us create a new path where the robot skips the cycle in the original path from $v$ through $t$ to $u$, then eventually back to $u$ through $t$ to $v$. This path must successfully reach the end, as every toggle after $t$ is traversed at most once, and so is in the same state regardless of whether the cycle is omitted.

Thus, under the assumption that there is a path which traverses toggle more than once, there is another, shorter path. Thus, the shortest path must not traverse toggles more than once, and so such a path must exist if any path exists.
\end{proof}
The remaining two-state two-tunnel deterministic reversible gadgets are also in P. We note that a wire-wire never changes its connectivity and is thus no different then two undirected edges. A lock-lock can never change its state and thus is reducible to a one state gadget, simply zero, one, or two undirected edges. A gadget with a tunnel which does not change and is not changed by the state of the gadget is reducible to two gadgets on one tunnel each, which are in P by Theorem~\ref{thm:1-toggle}. This exhausts the 2-state 2-tunnel reversible undirected gadgets.

\begin{proof}[Proof of Theorem~\ref{thm:general}]
    Now, we can characterize all two state, deterministic, reversible gadgets on any number of tunnels.

    Any gadget with two toggle tunnels, a toggle tunnel and a tripwire tunnel, a toggle tunnel and a lock tunnel or a tripwire and a lock tunnel is sufficient to make motion planning hard, by ignoring all other tunnels and using one of the constructions from this paper.

    We can divide all other gadgets into three categories: those with tripwires and trivial tunnels, those with locks and trivial tunnels, and those with a single toggle and trivial tunnels. The passability of a tunnel in a gadget with only tripwires and trivial tunnels never changes, making motion planning equivalent to st-connectivity. A gadget with only locks and trivial tunnels can never have its state change, allowing us to apply Theorem~\ref{thm:1-state}. A gadget with a single toggle and some number of trivial tunnels can be treated as a one-toggle together with some number of undirected edges. Thus, any system of gadgets of these types is equivalent to a system of 1-toggles and undirected edges. After that, the same argument as in Theorem~\ref{thm:1-toggle} can be used to solve the motion planning problem in that system.
\end{proof}

\section{Open Problems / Conclusion}
This framework for abstract motion planning problems leaves open the question of the computational complexity of motion planning with many other types of gadgets. One can examine gadgets with more states, without the tunnel restriction, or without the deterministic and reversible restrictions. Since this is a vast undertaking with many of the gadgets and their combinations likely to be uninteresting, we suggest some of the following categories to be of particular interest.

\begin{itemize}
\item 3 spinners are the only size of spinner for which motion planning remains open.
\item Three location, 2-state, deterministic, reversible gadgets seem like the obvious `simplest' category of gadgets.
\item Are there any sets of purely deterministic and reversible gadgets for which motion planning is PSPACE-complete (e.g. without branching hallways, which are non-deterministic)?
\item What about reversible but nondeterministic gadgets on two tunnels or three locations?
\end{itemize}

There is currently significant partial progress on all of the listed topics. Please contact us before spending significant time working on the open problems listed to prevent duplication of effort.

\section*{Acknowledgments}
This work grew out of an open problem session from the MIT class on
Algorithmic Lower Bounds: Fun with Hardness Proofs (6.890) from Fall 2014.
We particularly thank Jeffrey Bosboom for posing the problem of analyzing
4-spinners from Legend of Zelda: Oracle of Seasons (in 2015),
for simplifying the 2-state $k$-tunnels proof, and for other helpful discussions.

\bibliographystyle{alpha}
\bibliography{PushPullBib}

\end{document}